\documentclass{amsart}[a4paper]

\makeatletter
\let\mymakefnmark\@makefnmark
\let\mythefnmark\@thefnmark

\newcommand{\restorefn}{\let\@makefnmark\mymakefnmark
\let\mythfnmakr\@thefnmark}
\makeatother

\usepackage{graphicx}%
\usepackage{multirow}%
\usepackage{amsmath,amssymb,amsfonts}%
\usepackage{amsthm}%
\usepackage{mathrsfs}%
\usepackage[title]{appendix}%
\usepackage{xcolor}%
\usepackage{textcomp}%
\usepackage{manyfoot}%
\usepackage{booktabs}%
\usepackage{listings}%
\usepackage{xspace}
\usepackage{float}
\usepackage[foot]{amsaddr}

\usepackage{cite}

\usepackage{hyperref}

\usepackage{textcomp}
\usepackage{xcolor}
\usepackage{cleveref}
\usepackage{comment}
\usepackage{csquotes}

\allowdisplaybreaks

\def\BibTeX{{\rm B\kern-.05em{\sc i\kern-.025em b}\kern-.08em
    T\kern-.1667em\lower.7ex\hbox{E}\kern-.125emX}}

\newcommand{\N}{\mathbb{N}}

\usepackage{mathtools} 
\newcommand{\defeq}{\vcentcolon=}
\usepackage{algorithm}
\usepackage{algorithmicx}%
\usepackage[noend]{algpseudocode}

\algnewcommand\algorithmicforeach{\textbf{for each}}
\algdef{S}[FOR]{ForEach}[1]{\algorithmicforeach\ #1\ \algorithmicdo}

\newtheorem{theorem}{Theorem}
\newtheorem{proposition}[theorem]{Proposition}%
\newtheorem{theorem*}{Theorem}
\newtheorem{lemma}[theorem]{Lemma}

\newcommand{\1}{\mathbf{1}}

\usepackage{adjustbox}
\usepackage{booktabs}
\usepackage[utf8]{inputenc}
\usepackage{csquotes}
\usepackage{enumerate}
\usepackage{multirow}
\usepackage[autolanguage]{numprint}
\usepackage{graphicx}
\usepackage[locale=US]{siunitx}
\usepackage{subfigure}
\usepackage{textcomp}
\usepackage{url}
\usepackage{tikz}
\usepackage{xspace}
\usepackage[colorinlistoftodos,prependcaption]{todonotes}

\definecolor{light-gray}{gray}{0.93}

\newcommand{\ie}{i.\,e.,\xspace}
\newcommand{\st}{s.\,t.,\xspace}

\newcommand{\etal}{\textit{et al.}\xspace}

\newcommand{\EE}[1]{\mathrm{\mathbf{E}}\left[#1\right]}
\newcommand{\PP}[1]{\mathrm{\mathbf{P}}\left(#1\right)}

\newcommand{\Oh}{\ensuremath{\mathcal{O}}}

\newcommand{\argmax}{\operatorname{argmax}}

\newcommand{\gain}[2]{\mathrm{gain}(#1,#2)}

\newcommand{\ceil}[1]{\left\lceil #1 \right\rceil}

\newcommand{\norm}[1]{\left\lVert#1\right\rVert}
\newcommand{\diag}[1]{\operatorname{diag}(#1)}
\newcommand{\diam}[1]{\operatorname{diam}(#1)}

\newcommand{\mat}[1]{\mathbf{#1}}
\newcommand{\myvec}[1]{\mathbf{#1}}
\newcommand{\ment}[3]{\mathbf{#1}[#2,#3]}
\newcommand{\vent}[2]{\myvec{#1}[#2]}
\newcommand{\Lap}{\mat{L}}
\newcommand{\Lpinv}{\mat{L}^\dagger}

\newcommand{\LpinvG}[1]{\mat{L}_{\mathnormal{#1}}^\dagger}

\newcommand{\ndG}[3]{\myvec{b_{#1}}(#2,#3)}
\newcommand{\shdG}[3]{\myvec{b}^2_{#1}(#2,#3)}
\newcommand{\effres}[2]{\myvec{r}(#1,#2)}
\newcommand{\effresG}[3]{\myvec{r}_{#1}(#2,#3)}
\newcommand{\trace}[1]{\operatorname{tr}(#1)}
\newcommand{\uvec}[1]{\mathbf{e}_{#1}}
\newcommand{\onesvec}{\mathbf{1}}

\newcommand{\numtrees}[4]{N_{#1,#4}(#2,#3)}

\newcommand{\extend}[1]{{#1}}

\newcommand{\egc}{e.\,g.,\xspace}
\newcommand{\iec}{i.\,e.,\xspace}

\usepackage{fancyhdr}

\renewcommand{\thefootnote}{\fnsymbol{footnote}}

\begin{document}

\title[Greedy Optimization of Resistance-based Graph Robustness with Global and Local Edge Insertions]
{Greedy Optimization of Resistance-based Graph Robustness with Global and Local Edge Insertions\restorefn\footnotemark
}

\author{Maria Predari}
\email{predarim@hu-berlin.de}
\address[Maria Predari, Lukas Berner, Henning Meyerhenke]{Department of Computer Science, Humboldt-Universität zu Berlin, Unter den Linden 6, 10099 Berlin, Germany}

\author{Lukas Berner}
\email{lukas.berner@hu-berlin.de}

\author{Robert Kooij}
\email{r.e.kooij@tudelft.nl}
\address[Robert Kooij]{Faculty of Electrical Engineering, Mathematics and Computer Science, Delft University of Technology, Mekelweg 4, 2628 CD, Delft, Netherlands}
\address[Robert Kooij]{UNIT ICT, Strategy \& Policy, TNO (Netherlands Organisation for Applied Scientific Research), P.O. Box 96800, 2509 JE, The Hague, Netherlands}

\author{Henning Meyerhenke}
\email{meyerhenke@hu-berlin.de}

\begin{abstract}

  The total effective resistance, also called the Kirchhoff index, provides a
  robustness measure for a graph $G$.
  \extend{
    We consider two optimization problems of adding $k$
    new edges to $G$ such that the resulting graph has minimal total effective resistance
    (i.\,e., is most robust) -- one where the new edges can be anywhere in the graph and one
    where the new edges need to be incident to a specified focus node.}
  The total effective resistance and effective resistances between nodes can be
  computed using the pseudoinverse of the graph Laplacian.
  The pseudoinverse may be computed explicitly via pseudoinversion;
  yet, this takes cubic time in practice and quadratic space.
  We instead exploit combinatorial and algebraic connections
  to speed up gain computations in an established generic greedy heuristic.
  Moreover, we leverage existing randomized techniques to boost the performance of
  our approaches by introducing a sub-sampling step.
  Our different graph- and matrix-based approaches are indeed significantly
  faster than the state-of-the-art greedy algorithm,
  while their quality remains reasonably high and is often quite close.
  Our experiments show that we can now process \extend{larger}
  graphs for which the application of
  the state-of-the-art greedy approach was \extend{impractical} before.

\end{abstract}

\maketitle

\stepcounter{footnote}\footnotetext{\scshape A preliminary version of this paper appeared in the Proc. of \textit{2022 IEEE/ACM International Conference on Advances in Social Networks Analysis and Mining} (ASONAM)~\cite{predari22faster}. We gratefully acknowledge support by German Research Foundation (DFG) project ALMACOM (grant ME 3619/4-1) and by the TU Delft Safety \& Security Institute project ARCIN.}
\setcounter{footnote}{0}

\renewcommand{\thefootnote}{\arabic{footnote}}

\section{Introduction}
\label{sec:intro}

The analysis of network topologies has received considerable attention
in various fields of science and engineering in the last decades~\cite{barabasi,Newman2018networks}.
Its purpose usually is to better understand the functionality,
dynamics, and evolution of a network\footnote{We use the terms
  \emph{network} and \emph{graph} interchangeably in this paper.}
and its components~\cite{barabasi}.
One important property of a network topology concerns its \emph{robustness},
\iec the extent to which a network is capable to withstand failures of one or more of
its components~\cite{freitas2022graph}.
As an example, one may ask whether the network is guaranteed to remain
connected if an edge is deleted, \egc due to failure or an attack.
Network robustness is a critical design issue in many areas,
including telecommunication~\cite{jose}, power grids \cite{yakup}, public transport~\cite{oded}, supply chains~\cite{supplychains} and water distribution~\cite{water}.

Often a critical step in infrastructural maintenance is to
improve the robustness of the network by adding a small
number of edges. The challenge here lies in
the selection of a vertex pair, among all the possible ones,
such that the insertion of an edge between the vertices increases the
network's robustness as much as possible.
Given a graph $G=(V,E)$ and a budget of $k$ links to be added, our algorithmic formalization
of this task asks to find a set $X \subset {V \choose 2} \setminus E$ of size $k$
that optimizes the robustness of $G$. We call this problem $k$-GRIP, short
for \emph{global robustness improvement problem}.
\extend{A related task fixes a focus node $v \in V$ from which $k$ edges can be inserted into $G$
  to other nodes; we call this problem $k$-LRIP, short for \emph{local robustness improvement problem}.}
Clearly, one must also choose a measure to
capture a sensible notion of robustness; there are numerous ones proposed in
the literature~\cite{barabasi,jose}.

One established measure for $k$-GRIP, which was shown to be a good robustness indicator
in various scenarios~\cite{ELLENS2011,ghosh,Wang2014ImprovingRO}, is \emph{effective graph resistance}
or \emph{total effective resistance} of a graph.
Effective resistance is a pairwise metric on the vertex set of $G$,
which results from viewing the graph as an electrical network.
It relates to uniform spanning trees~\cite{apgm20}, random walks~\cite{Lovsz96RandomWO},
and several centrality measures~\cite{Mavroforakis15, BrandesF05}.
\extend{In fact, it works similarly as an objective function for $k$-LRIP -- we are just restricted in
  the search space to a particular focus node.}
To compute the total effective resistance, one sums the effective resistance over
all vertex pairs in $G$ (for technical details see Section~\ref{sec:prelim}).
Intuitively, the effective resistance becomes small if there are many short paths
between two vertices. Removing an edge in such a case hardly disrupts
the connectivity, since there are usually alternative paths.
Due to this favorable property, we select total effective resistance in this
paper as the robustness measure for $k$-GRIP \extend{and $k$-LRIP}.

The effective graph resistance-based $k$-GRIP version, \extend{recently shown to be $\mathcal {NP}$ hard~\cite{kooij2023resistanceNP},}
was already considered by Summers \etal~\cite{top15}. It was shown in~\cite{summerscorrection} that $k$-GRIP for the effective graph resistance is not submodular, hence without an approximation guarantee for the greedy algorithm (more details in Section~\ref{sec:related}). Still, even without an approximation guarantee, this greedy algorithm provides
very good empirical results -- for small networks it does so in reasonable time. It should be noted that the example given in~\cite{summerscorrection} which proves that $k$-GRIP for the effective graph resistance is non-submodular, also proves that $k$-LRIP is non-submodular for the effective graph resistance.
\newline
The greedy algorithm performs $k$ iterations, at each step adding the edge with highest marginal gain.
To compute these gains, however, the corresponding effective resistance
values are needed. If one acquires them by an initial (pseudo)inversion of
the graph's Laplacian matrix, this takes $\Oh(n^3)$ time with standard tools
in practice (where $n = |V|$). Overall, this approach leads to a running time
of $\Oh(k n^3)$, which limits the applicability to large networks. 

For other problems where this greedy approach works well, a recent stochastic
greedy algorithm~\cite{Mirzasoleiman2015LazierTL} has been shown to be potentially much faster --
while usually producing solutions of nearly the same quality. It does so
by \emph{sampling} from the set of candidates to find the one with highest gain
(from the sample instead of from the whole set) in each iteration.
Our hypothesis for this paper is that this favorable
speed-quality tradeoff of stochastic greedy holds for our $k$-GRIP as well.
We also assume that other Laplacian approximation techniques
can speed up the required computations. Furthermore, we hope that the techniques
that work well for the $k$-GRIP problem also work well (if adapted properly) for the
related $k$-LRIP problem. Some differences in the speed-quality tradeoff are to be expected.

Building upon the generic stochastic greedy approach~\cite{Mirzasoleiman2015LazierTL},
we first devise several heuristic strategies for $k$-GRIP that leverage both graph- and matrix-related properties (Section~\ref{sec:contrib}).
Our approaches accelerate the greedy algorithm by reducing the candidate set
via careful selection of elements to be evaluated and/or by accelerating the
gain computation.
Our experiments (Section~\ref{sec:expes}) confirm that our approaches speed up
the state-of-the-art greedy algorithm significantly. At the same time, the $k$-GRIP solution quality is
more or less preserved, how well depends on the approach. For instance, testing graphs
with $< 57K$ nodes, we produce results that are on average $2-15$\% away from the
greedy solution, while running $3.3-68\times$ faster than the state of the art (SotA).
Finally, we demonstrate that we can now
process much larger graphs for which the application of
the SotA greedy approach was infeasible before.

\extend{Besides a better update strategy for our heuristic \textsc{ColStoch},
  another extension of this paper compared to its conferece version~\cite{predari22faster} consists of
  the $k$-LRIP part (Section~\ref{sec:lrip}). The corresponding experiments in Section~\ref{sec:expes} show
  that our heuristics (except one) work for this problem similarly well when the graphs are sufficiently large.
  For example, on graphs with more than \numprint{10000} nodes,
  one of our new heuristics is $\approx 10\%$ away from the greedy quality,
  but on average $\approx 2$-$7\times$ faster (depending on $k$ and the graph).
}

\section{Preliminaries}
\label{sec:prelim}
%
We assume that our input consists first of all of a connected,
undirected, and simple graph $G = (V, E)$ with $n$ vertices 
and $m$ edges. 
For both $k$-GRIP and $k$-LRIP, we also have an integer $k\in\mathbb{Z}_{>0}$ for the number
of edges to be added to $G$; $k$-LRIP additionally requires the focus node $v \in V$ from which
the additional edges are inserted. Our methods can be easily extended to weighted graphs.
However, for sake of presentation simplicity, we only consider unweighted graphs.

For the remainder we use several well-known matrix representations of graphs.
$\mat{L} = \mat{D} - \mat{A}$ is the $n\times n$ Laplacian matrix of $G$,
where $\mat{D}$ is the diagonal matrix of vertex degrees and $\mat{A}$ the adjacency matrix.
$\mat{L}$ is symmetric, positive semi-definite and has
zero row/column sum $\st$ $\Lap \onesvec = \textbf{0}$ where $\onesvec$ is
the all-ones vector.
The $m\times n$ incidence matrix $\mat{B}$ takes for $e \in E$ and $a \in V$
the values: $\ment{\mat{B}}{e}{a} = 1$ if $a$ is the
destination of $e$, $\ment{\mat{B}}{e}{a} = -1$ if
$a$ is the origin of $e$ and $\ment{\mat{B}}{e}{a} = 0$ otherwise.
For undirected graphs, the direction of each edge is specified arbitrarily.
Moreover, $\mat{L} = \mat{B^T}\mat{B}$.
It is well-known that $\Lap$ is not invertible,
so that its Moore-Penrose
pseudoinverse ($\Lpinv$) is used instead, for which holds:
$\Lap \Lpinv = \Lpinv \Lap = \mat{I} - \frac{1}{n}\cdot \onesvec \onesvec^T$~\cite{Gutman04}.
Since $\Lap$ is symmetric, it has an orthonormal basis of eigenvectors
$\mat{U} = [\myvec{u_1},\ldots, \myvec{u_n}]$. We write the spectral decomposition
as: $\Lap = \sum_{i=2}^{n} \myvec{u_i} \lambda_i \myvec{u_i}^T$, where the eigenvectors
$\myvec{u_2},\ldots, \myvec{u_n}$ correspond to the ordered eigenvalues
$ 0 < \lambda_2 \leq, \ldots, \leq \lambda_n$ (excluding the zero eigenvalue).

For a graph $G$ we use $\mat{L}_{G}$ [$\LpinvG{G}$] to refer to its
Laplacian [Laplacian pseudoinverse].
If there is no subscript in our matrix notation, the associated graph is
inferred by the context.

Let $\Omega_G := {V \choose 2} \setminus E $.
For any $X \subset \Omega_G$, we define $G' \defeq G \cup X = (V, E \cup X)$
as the graph obtained by adding the edges of $X$ into $G$.
Then, $k$-GRIP aims at finding a set $X \subset \Omega_G$ with $|X| = k$
\st $|f(G) - f(G')|$ is as large as possible
for a given robustness function $f(\cdot)$.
Here, we use the effective graph
resistance $\mathcal R(G)$ as robustness function (for which lower values indicate higher robustness), which
is the sum of pairwise effective resistances $\effresG{G}{\cdot}{\cdot}$
between all vertex pairs:
\begin{equation}
  \label{eq:rtot}
  \mathcal R(G) = \sum_{a=1}^{n}\sum_{b=a+1}^{n} \effresG{G}{a}{b} \,.
\end{equation}

\extend{Thus, $k$-GRIP for total effective resistance asks to find the set $X$ of size $k$
  that minimizes the resistance of the graph resulting from inserting the edges of $X$.}
\extend{
  The notion of effective resistance comes from viewing $G$
  as an electrical circuit in which each edge $e$ is a resistor
  with resistance $1/\vent{w}{e}$.
  Following fundamental electrical laws, the effective resistance $\effres{a}{b}$
  between two vertices $a$ and $b$ is the potential difference
  between $a$ and $b$ when a unit current is injected into
  $G$ at $a$ and extracted at $b$.
}

\extend{
  The second problem we address is the related $k$-LRIP problem. It also uses total effective
  resistance $\mathcal{R}(G)$ as the objective function. The main difference is that it restricts
  the search space by limiting the insertion of the $k$ edges to a particular focus node
  $v \in V$ that is part of the input. The set $X$ of edges to insert is selected from
  the vertex pairs $\Omega_v := \{(v,u) ~|~ u \in V, \{v,u\} \notin E\}$.
}

Computing $\effresG{G}{a}{b}$ can be done via $\Lpinv$:
\begin{equation}
  \label{eq:effres}
  \effresG{G}{a}{b} = \ment{\Lpinv}{a}{a} + \ment{\Lpinv}{b}{b} - 2 \ment{\Lpinv}{a}{b} \,.
\end{equation}

Combining Eqs.~(\ref{eq:rtot}) and~(\ref{eq:effres}), one gets
\begin{equation}
  \label{eq:rtot-trace}
  \mathcal R(G) = n\trace{\Lpinv} \,.
\end{equation}

For a potential new edge $\{a,b\}$, we have $G' = G \cup \{a,b\}$ and
$\mat{L}_{G'} = \mat{L}_{G} + (\myvec{e}_a - \myvec{e}_b )(\myvec{e}_a - \myvec{e}_b)^T$,
where $\myvec{e}_a$ is a zero vector except for $\vent{e}{a} = 1$.
The gain in terms of $\mathcal{R}$ by the insertion of $\{a,b\}$ is $\mathcal R(G) - \mathcal R(G')$
and relies on $\LpinvG{G'}$ (Sherman-Morrison formula~\cite{SMFormula}):
\begin{align}
  \LpinvG{G'} & = \LpinvG{G} - \frac{1}{1 + \effresG{G}{a}{b}} \LpinvG{G}(\myvec{e}_a - \myvec{e}_b )(\myvec{e}_a - \myvec{e}_b)^T \LpinvG{G}\,. \label{eq:lpinv_update}
\end{align}
The gain evaluation $\gain{a}{b} = \mathcal R(G) - \mathcal R(G')$ is then
\begin{equation}
  \label{eq:gain}
  \gain{a}{b} = n \frac{\norm{\ment{\LpinvG{G}}{:}{a} - \ment{\LpinvG{G}}{:}{b}}^2}{1+ \effresG{G}{a}{b} }  \,,\\
\end{equation} where $\ment{\LpinvG{G}}{:}{i}$ is the $i^{th}$ column of $\Lpinv$. We rewrite Eq.~(\ref{eq:gain}) as a function of squared $\ell_2$ norms:
\begin{equation}
  \label{eq:gain-jlt}
  \gain{a}{b} = n \frac{\norm{\Lpinv(\myvec{e}_a - \myvec{e}_b)}^2} {1 + \norm{\mat{B}\Lpinv(\myvec{e}_a - \myvec{e}_b)}^2} = n\frac{\shdG{G}{a}{b}}{1+ \effresG{G}{a}{b}}\,,
\end{equation}
where $\ndG{G}{\cdot}{\cdot}$ is known as the biharmonic distance of $G$~\cite{yi2018biharmonic,Wei2021BiharmonicDO}. Finally, we express these distances via the spectral decomposition of $\Lpinv$ (or $\mat{L}$, respectively):
\begin{equation}
  \label{eq:denominator}
  \begin{split}
    & \effresG{G}{a}{b} = \norm{\mat{B}\Lpinv(\uvec{a} - \uvec{b})}^2 = (\uvec{a} - \uvec{b})^T\Lpinv(\uvec{a} - \uvec{b}) \\
    & = (\uvec{a} - \uvec{b})^T \mat{U} \mat{\Lambda}^{-1} \mat{U}^T(\uvec{a} - \uvec{b}) =  \sum_{i = 2}^{n} \frac{(\vent{u_i}{a} - \vent{u_i}{b})^2}{\lambda_i}  \,,
  \end{split}
\end{equation}

where $\mat{\Lambda}$ is the diagonal matrix of the eigenvalues of $\mat{\Lpinv}$. Similarly:
\begin{equation}
  \label{eq:nominator}
  \begin{split}
    & \shdG{G}{a}{b} = \norm{\Lpinv(\uvec{a} - \uvec{b})}^2  = (\uvec{a} - \uvec{b})^T(\Lpinv)^2(\uvec{a} - \uvec{b}) \\
    & = (\uvec{a} - \uvec{b})^T \mat{U} \mat{\Lambda}^{-2} \mat{U}^T(\uvec{a} - \uvec{b})
    = \sum_{i = 2}^{n} \frac{(\vent{u_i}{a} - \vent{u_i}{b})^2}{\lambda_i^2} \,.
  \end{split}
\end{equation}

\section{Related Work}
\label{sec:related}
Robustness of networks has been an active research area for decades~\cite{ps18,freitas2022graph}. Several authors have proposed the use of specific network metrics to quantify the robustness of a given network, see e.\,g.,~\cite{jose}, \cite{fiedler}, \cite{schneider}, \cite{hale}. \extend{In a recent survey on the topic, Freitas \etal~\cite{freitas2022graph} classify robustness metrics into three types:
  metrics based on structural properties, such as edge connectivity or diameter; metrics based on the spectrum of the adjacency matrix, such as the spectral radius or spectral gap; and metrics based on the spectrum of the Laplacian matrix, for instance the algebraic connectivity and the effective graph resistance.}
Here, the algebraic connectivity,
\ie the second smallest eigenvalue $\lambda_2$ of the graph's Laplacian~\cite{fiedler}, is known to capture the overall connectivity of a graph. This metric is also related to synchronization of networks, including opinion dynamics~\cite{consensus}.

\extend{
  Once the robustness of a network has been established, a natural next step is to determine how robustness can be improved. Schneider \etal~\cite{schneider} view the relative size of the largest connected component as robustness measure
  (after removing a certain fraction of the edges) and rewire the edges for robustness improvement.}
A second approach is to add elements to the network.
\extend{Several researchers investigated $k$-GRIP for specific robustness metrics. For instance, Ref.~\cite{wang2008algebraic} considered $1$-GRIP, with the robustness metric being the algebraic connectivity. They suggest several strategies, based upon topological and spectral properties of the graph, to decide which single link to add to the network in order to increase the algebraic connectivity as much as possible.
  Ref.~\cite[Chapter 8]{ZhidongThesis} also considered algebraic connectivity for $k$-GRIP. Under some light conditions, lower bounds for the quality of the greedy solution were obtained.
  It might be argued that the algebraic connectivity is not a proper robustness metric, because there are examples where adding a link to a graph does not change the algebraic connectivity, see~\cite{natural_connectivity}.
  The $\mathcal{NP}$-hardness of $k$-GRIP for algebraic connectivity was proved in \cite{AC_NP}.}
\extend{Manghiuc \etal~\cite{DBLP:conf/esa/Manghiuc0S20} consider a weighted decision variant of $k$-GRIP w.r.t. $\lambda_2$. They propose an almost-linear time algorithm that augments the graph by $k$ edges such that $\lambda_2$ exceeds a specified threshold. A nice overview of algebraic connectivity for $k$-GRIP is presented in Ref.~\cite{li2018algebraicconnectivity}.}

\extend{Papagelis~\cite{papagelis2015averagepathlength} shows that $k$-GRIP with the average shortest path length as a robustness metric does not satisfy the submodularity constraint, but accurate greedy solutions can be obtained. Van Mieghem \etal~\cite{vanmieghem2011spectralradius} consider a link removal problem with the spectral radius (largest eigenvalue of adjacency matrix) as a robustness metric and prove this problem is $\mathcal{NP}$-hard. Baras and Hovareshti~\cite{baras2009spanningtrees} consider the problem of adding $k$ links to a given network, such that the number of spanning trees in the graph is maximized.}

Effective graph resistance as a robustness measure dates back at least to
Ellens \etal~\cite{ELLENS2011}. It has been known much longer, however, that
effective resistance is proportional to commute times of random
walks~\cite{ghosh}.
Refs.~\cite{Wang2014ImprovingRO} and~\cite{ps18}
investigate heuristics for $1$-GRIP with effective graph resistance (both for
edge insertion \emph{and} removal).
Besides deriving theoretical bounds, Wang \etal~\cite{Wang2014ImprovingRO}
compare spectral strategies for edge selection with much simpler heuristics.
Their experiments confirm that their spectral strategies
(particularly the one based on the highest effective
resistance gain) often yield the largest improvement, indicating
a tradeoff between running time and the robustness gain.

\extend{Pizzuti and Socievole~\cite{ps18,Pizzuti2023} proposed and evaluated several genetic algorithms to find the optimal edge to add, in order to minimize $R_G$. Clemente \textit{et al.}~\cite{clemente2020graphresistance} studied $k$-GRIP for the effective graph resistance and gave lower bounds for $R_G$ upon the addition of $k$ links, under some mild conditions for $k$. For $k = 1$ the lower bound in \cite{clemente2020graphresistance} clearly outperforms the lower bound in \cite{Wang2014ImprovingRO}.}

The state-of-the-art heuristic for $k$-GRIP
is a greedy algorithm presented by Summers \etal~\cite{top15},
called here {\sc StGreedy}. In its generic form, such a greedy
algorithm adds in each of the $k$ iterations
the element (here: edge) with the largest marginal gain
(here: best improvement of the robustness measure).
To this end, {\sc StGreedy} computes the full
pseudoinverse of $\mat{L}$ as a preprocessing step.
Then, the marginal gains of all vertex pairs
are computed via Eq.~(\ref{eq:gain}) in $\Oh(n)$ time per edge.
The edge with best marginal gain is added to the graph, and
the pseudoinverse is updated 
using Eq.~(\ref{eq:lpinv_update}).
The time complexity is $\Oh(k n^3)$, which is due to the evaluation of
the gain function in $k$ rounds on $\Oh(n^2)$ node pairs.
The preprocessing takes $\Oh(n^3)$ time with standard tools.
For monotonic submodular problems, the generic greedy algorithm has
an approximation ratio of
$1-1/e$. Even for non-submodular problems such as $k$-GRIP
(see~\cite{summerscorrection} for a counterexample),
the greedy algorithm still often leads to solutions of
high quality~\cite{DBLP:conf/eucc/SummersK19,DBLP:conf/alenex/AngrimanBDGGM21}.
Stochastic greedy algorithms that improve the time complexity of the standard greedy approach
(in a general setting) were proposed in Refs.~\cite{Mirzasoleiman2015LazierTL, hassidim17a}.
These algorithms use random sampling techniques and reduce the total number of function
evaluations (roughly) by a factor of $k$. They achieve provable approximation guarantees
in cases where the greedy algorithm admits them, too.

\extend{Also, $k$-LRIP has been considered by several authors -- for different objectives.
  Shan~\etal~\cite{shan2018noderesistance} consider the node resistance (as robustness metric
  or rather as a centrality measure), which is the sum of the effective resistance from one source node $v$ to all other nodes.
  They assume that the $k$ links that are to be added are chosen from the set of non-existing links from the focus node $v$; not all possible non-existing links. It is shown by the authors that in this setting, the node resistance is a supermodular set function.
  Ref.~\cite{Bergamini} considers $k$-LRIP with betweenness centrality. In fact, $k$-LRIP has been studied with a variety of other centrality metrics, such as PageRank~\cite{avrachenkov2006effect},~\cite{olsen2014approximability}, closeness centrality~\cite{crescenzi2016greedily} and eccentricity~\cite{demaine2010minimizing},~\cite{perumal2013minimizing}. }

\extend{
  Besides using the stochastic greedy algorithm for both $k$-GRIP and $k$-LRIP,
  we intend to accelerate the optimization
  process by approximation techniques for the effective resistance values.
  While Shan~\etal~\cite{shan2018noderesistance} also employ a greedy algorithm for $k$-LRIP,
  their objective function and acceleration techniques differ from ours.
}

\section{Heuristics for $k$-GRIP}
\label{sec:contrib}
In this section, we propose different techniques to improve the performance
of the greedy algorithm for $k$-GRIP. Our approaches are:
{\sc SimplStoch}, {\sc ColStoch}, {\sc SimplStochJLT}, { \sc ColStochJLT} and {\sc SpecStoch}.
They all make use of existing randomized techniques
and follow the general greedy framework of Algorithm~\ref{alg:stochastic}.
Functions named as \textsc{Obj*} relate to the objective function
while those named as \textsc{Candidate*} relate to the set of possible candidate
elements. Functions not defined explicitly in the pseudocode are described in detail in the text.
\extend{The time and space complexities of all approaches (assuming standard tools)} are shown in Table~\ref{tbl:complexities}.

\begin{algorithm}[H]
  \begin{algorithmic}[1]
    \begin{small}
      \Function{GreedyFramework}{$G$, $k$, $\delta$}
      \State \textbf{Input:} Graph $G=(V,E)$, $k \in \N_{>0}$, accuracy $ 0<\delta <1$
      \State \textbf{Output:} $G_k$ -- graph after $k$ edge insertions

      \State $G_0 \gets  G$
      \State  \textsc{ComputeObj}($G_0$, $\dots$) \Comment {\small compute step} \label{line:compute}
      \State $s \gets$ \textsc{CandidateSize}($m$, $n$, $k$, $\delta$)    \label{line:candidate-size}
      %
      \For {$r \gets 0, \dots, k-1$}    \label{line:main-loop}  \Comment {\small main loop}
      \State $\mathcal S \gets$ \textsc{Candidates}($s$, $G_r$, $\ldots$) \label{line:candidate}
      \ForEach {$\{a,b\} \in \mathcal S \times \mathcal S$} \Comment {\small \# of evaluations} \label{line:candidate-iteration}
      \State $\gain{a}{b} \gets$ \textsc{Eval}($a$, $b$, $\dots$) \Comment {\small single evaluation} \label{line:single-eval}
      \EndFor
      \State $(a^*,b^*) \gets \argmax_{a\in S \times b\in S} \gain{a}{b}$
      \State $G_{r+1} = G_r \cup (a^*,b^*) $ \label{line:graph-update}
      \State  \textsc{Update}($G_{r+1}$, $\dots$)  \Comment {update step} \label{line:update}
      \EndFor \label{line:end-main-loop}
      \State \textbf{return} $G_{r+1}$
      \EndFunction
    \end{small}
  \end{algorithmic}
  \caption{General framework for $k$-GRIP}
  \label{alg:stochastic}
\end{algorithm}

For submodular functions the greedy framework can be combined with a lazy
technique~\cite{Minoux78} that boosts the performance of the algorithm.
This process is based on the fact that, even though marginal gains
of elements might change between iterations, their order often stays the same.
An observation important for us is: \enquote{(T)he lazy greedy algorithm can
  be applied to cases with no strict guarantee (for submodularity) since
  experience shows that it most often produces the same final solution
  as the standard greedy algorithm}~\cite{Minoux89}. Based on the above
observation and existing, positive results on the lazy greedy approach
for $k$-GRIP~\cite{top15}, we also employ this technique
\extend{and do so by means of a priority queue.
  Entries in the priority queue are of
  the form $(e, g(e), r)$, where $e\in {V \choose 2} \setminus E$,
  $g(e)$ is the marginal gain of $e$, and $r\in \N_{>0}$ is the round in which the gain was computed.}

All our approaches improve the speed of the greedy algorithm by
reducing the candidate set and/or by accelerating the objective function calculation/update.
Nearly inevitably, the above incurs a smaller or larger trade-off
between speed improvement and solution quality degradation.

\begin{table}[H]
  \caption{Time complexities \extend{(assuming standard (pseudo)inversion tools, linear solvers, and eigensolvers used in practice for Laplacians of general graphs)}
    of all approaches involved. Columns correspond to major steps of Algorithm~\ref{alg:stochastic}.
    In general, the dominant term comes from the total number of evaluations and their time to be evaluated (second column).
    \extend{The $\widetilde{\ensuremath{\mathcal{O}}}$-notation hides $\log(1/\epsilon)$ factors, where $\epsilon$
      is the accuracy threshold of the linear solver. The ${\mathcal{O}}'$-notation
      hides $\log(1/\delta)$ factors, where $\delta$ determines the sample size in the stochastic candidate selection. Note that we consider the Johnson-Lindenstrauss transform (JLT) parameter $\eta$ here as a constant.
      $\tau$ is the number of uniform spanning trees (USTs) required for the diagonal approximation in \textsc{ColStoch}, which depends on the diameter of the graph~\cite{apgm20}. More details in the text.}}
  \centering

  \resizebox{\columnwidth}{!}{%
    \begin{tabular}{lllll} \toprule
                          & Compute                                                & \#Evals $\times$ SingleEval        & \extend{All updates}                                    & \extend{Memory}              \\
      \midrule
      {\sc StGreedy}      & $\Oh(n^3)$                                             & $\Oh(kn^2)$ $\times$ $\Oh(n)$      & $\Oh(kn^2)$                                             & $\Oh(n^2)$                   \\
      {\sc SimplStoch}    & $\Oh(n^3)$                                             & $\Oh'(n^2)$ $\times$ $\Oh(n)$      & $\Oh(kn^2)$                                             & $\Oh(n^2)$                   \\
      {\sc ColStoch}      & $ \widetilde{\ensuremath{\mathcal{O}}}(sm\log n )$     & $\Oh'(n^2)$ $\times$ $\Oh(n)$      & $ \widetilde{\ensuremath{\mathcal{O}}}(ksm\log n )$     & $\Oh((s+\tau)n+m)$           \\
      {\sc SimplStochJLT} & $\widetilde{\ensuremath{\mathcal{O}}}(m\log n)$        & $\Oh'(n^2)$ $\times$ $\Oh(\log n)$ & $\widetilde{\ensuremath{\mathcal{O}}}(km\log n)$        & $\Oh((s+\log n)n+m)$         \\
      {\sc ColStochJLT}   & $\widetilde{\ensuremath{\mathcal{O}}}(m\log n \log s)$ & $\Oh'(n^2)$ $\times$ $\Oh(\log s)$ & $\widetilde{\ensuremath{\mathcal{O}}}(km\log n \log s)$ & $\Oh((\log s + \tau) n + m)$ \\
      {\sc SpecStoch}     & $\Oh(c m)$                                             & $\Oh'(n^2)$ $\times$ $\Oh(c)$      & $\Oh(k c m)$                                            & $\Oh(c n + m)$               \\
      \bottomrule\end{tabular}
  }
  \label{tbl:complexities}
\end{table}

\subsection{{\sc SimplStoch}}
\label{sec_stochastic_greedy}

Our first idea is to simply apply the generic randomized technique
proposed in generic form by Mirzasoleiman \etal~\cite{Mirzasoleiman2015LazierTL}
in the context of $k$-GRIP.
The main idea of Ref.~\cite{Mirzasoleiman2015LazierTL} is to not inspect all
possible elements for insertion,
but only a reduced sample $\mathcal S$. For non-negative monotone submodular
functions (which does not hold for $k$-GRIP), the stochastic greedy approach provides
an approximation ratio of  $1 - e^{-(1-\delta)}$, where $ 0\leq \delta \leq 1$
is an accuracy parameter.

Regarding {\sc SimplStoch}, any edge \extend{from $\mathcal{S} \times \mathcal{S}$
  is a subset of ${V \choose 2} \setminus E$;}
during each iteration of the main loop we sample uniformly at random
$s \defeq \frac {n(n-1)/2-m}k \log{(\frac 1 {\delta})}$ vertex pairs (Line~\ref{line:candidate} in Algorithm~(\ref{alg:stochastic})),
resulting in $\Oh((n^2-m) \log{(\frac 1{\delta})})$ function evaluations overall.
Those are performed via the Laplacian pseudoinverse obtained during preprocessing, in a similar
way as in {\sc StGreedy}.
More precisely, $\Lpinv$ is computed once before the main loop
(Line~\ref{line:compute}) and is used within the loop to quickly
determine single evaluations (Line~\ref{line:single-eval}). Every time an edge is added to the graph,
$\Lpinv$ is updated accordingly via Eq.~(\ref{eq:lpinv_update})
(Line~\ref{line:update}). The cost of the main loop for {\sc SimplStoch} is reduced compared to greedy
by a factor of $k / \log(1/\delta)$.
Yet, computing $\Lpinv$ is still very time- and space-consuming.

\subsection{{\sc ColStoch}}
Our first improvement upon {\sc SimplStoch} 
avoids the full pseudoinversion of $\Lap$, reducing the cost of Line~\ref{line:compute}
in Alg.~\ref{alg:stochastic}.
To this end, we make the following observation:
evaluating a single vertex pair $\{a,b\}$ via Eq.~(\ref{eq:gain})
requires only two columns of $\Lpinv$; precisely those corresponding
to vertices $a$ and $b$. That is why, instead of sampling elements
from ${V \choose 2} \setminus E$, {\sc{ColStoch}} restricts the sampling
process to elements from $V$, the set of columns of $\Lpinv$.
Carefully selecting $\mathcal S$ is critical as it affects the quality
of the solution.
Even if our problem is not submodular, we choose the default sample size of
$s  = n \sqrt {\frac 1 {k} \cdot\log(\frac 1 {\delta})}$
elements (Line~\ref{line:candidate-size}), leading to $\Oh(n^2 \log (\frac 1{\delta}))$
evaluations over all iterations, similar to {\sc SimplStoch}.
The only difference here is that we sample pairs of $\Lpinv$ columns, which is a subset of
${V \choose 2}$ and not ${V \choose 2} \setminus E$. Obviously, we reject
vertex pairs that already exist in the graph as edges.

Moreover, to limit the quality loss, we choose
elements of $\mathcal S$ following graph-based sampling
probabilities (details in Section~\ref{para:strategy}).
These probabilities are initially calculated during
the compute step (Line~\ref{line:compute}) and are updated accordingly
in the update step (Line~\ref{line:update}). Function \textsc{Candidates}()
also receives those sampling probabilities in each iteration (Line~\ref{line:candidate}).
Once $\mathcal S$ is determined, we compute all columns
of $\Lpinv$ corresponding to vertices in $\mathcal S$.
This step is performed once in the main loop after Line~\ref{line:candidate}. For the complexity analysis
we consider it as part of the compute step and for that reason it is not depicted in the loop of the generic Algorithm~\ref{alg:stochastic}.

We compute the columns corresponding to $\mathcal S$ by solving $s$ linear systems.
More precisely,
we solve one linear system for each vertex
$a \in \mathcal S: \mat{L}\myvec{x} = \uvec{a} - \frac{1}{n}\cdot \onesvec$,
where $\onesvec = (1, \dots, 1)^T$ and $\myvec{x} \perp \onesvec$.
\extend{Once the sample set $\mathcal S \subset V$ is determined, }
{\sc ColStoch} performs function evaluations only
between vertex pairs in $ \mathcal S\times \mathcal S$ (Line~\ref{line:single-eval}).
Finally, to further improve the overall running time,
we do not update $\ment{\LpinvG{G}}{:}{\mathcal S}$ for all $a \in \mathcal S$ at
the end of each round (Line~\ref{line:update} of Algorithm~\ref{alg:stochastic}).
Instead, we update individual columns of $\Lpinv$ on demand; only if the corresponding
vertices participate in the candidate set $\mathcal S$ of the following round.

To update previously computed columns, we use the outdated solver solution and apply
the update formula Eq.~(\ref{eq:lpinv_update}) iteratively
for all (in-between) rounds. To do so, we store columns together with the associated
round number.

\subsubsection{$\diag{\Lpinv}$ Strategy}
\label{para:strategy}
Let us now explain the sampling probabilities for selecting $\mathcal S$.
Following previous studies~\cite{VMieghem17, Wang2014ImprovingRO},
vertex pairs with maximal effective resistance are good candidates
for largely decreasing the total effective resistance of a graph.
However, the effective resistance metric is not directly applicable in our immediate context.
Firstly, because {\sc ColStoch} requires a vertex-based metric and secondly (and more importantly)
because computing the effective resistance for all vertex pairs $\{a,b\} \in
  {V \choose 2} \setminus E$ would eventually mean to (pseudo)invert $\Lap$ -- with the associated cost.
To circumvent these issues, we sample vertices according to their corresponding
diagonal entries in $\Lpinv$. \extend{Recall from Section~\ref{sec:prelim} that these entries are
  proportionate to the electrical farness of the corresponding nodes.
  In other words,} the diagonal entry $\ment{\Lpinv}{a}{a}$ of a vertex $a$
corresponds to the summed effective resistance between
$a$ and all other vertices: $\sum_{b\in V \setminus \{a\}} \effresG{G}{a}{b}$.
Vertices with maximum $\Lpinv$ diagonal values are
connected badly to all other vertices in the graph (in the electrical sense)~\cite{VMieghem17},
\extend{which is why we select them with higher probability for an edge insertion}.

\extend{
  Computing $\diag{\Lpinv}$ can be performed in almost-linear time by
  using the connection of effective resistance to uniform spanning trees (USTs) of $G$.
  A UST of $G$ is a spanning tree drawn uniformly at random from the set of all spanning trees of $G$.
  Angriman \etal~\cite{apgm20} proposed an algorithm that approximates (effective resistances and)
  $\diag{\Lpinv}$ via UST sampling techniques. The algorithm obtains a
  $\pm \epsilon$-approximation with high probability in $\Oh(m\log^4 n \cdot \epsilon^{-2})$ time for small-world graphs
  (diameter bounded by $\Oh(\log n)$). We provide here some details necessary to understand our new update
  strategy (Section~\ref{sub:update-lpinv}) when an edge is added.
}

\extend{
  Following fundamental electrical laws, the effective resistance $\effres{u}{v}$ of
  vertices $u$ and $v$ is the potential difference between $u$ and $v$ when a unit
  of current is injected into $G$ at $u$ and extracted at $v$.
  According to Ohm's law, whenever there is a potential vector $\myvec{x} \in \mathbb{R}^{n\times 1}$ on the vertices of $G$,
  there is also an electrical flow $\myvec{f} \in \mathbb{R}^{m \times 1}$ on the edges of the graph, equal to the potential differences
  and leading from the node with higher to the node with lower potential value.
  As a consequence, we can express $\effres{u}{v}$ (for any vertex pair $(u,v)$) as the sum of current flows on
  \emph{any} path\footnote{For the algorithm, it is beneficial to use shortest paths, though.} $\langle u = v_0, v_1, \dots, v_{k-1}, v_k = v \rangle$ as:
  \begingroup
  \begin{align}
    \label{eq:flow-path}
    \effres{u}{v} & = \sum_{i=0}^{k-1} \vent{f}{v_i, v_{i+1}}
  \end{align}
  \endgroup
}
\extend{
  Note that the sign of the current flow changes if we traverse an edge against the
  flow direction (and thus the sum may hide negative values when the direction is reversed).
  Eq.~(\ref{eq:flow-path}) can also be written as~\cite{bol98}
  \begingroup
  \begin{align}
    \label{eq:N-values-in-path}
    \effres{u}{v} & = 1/N \sum_{i=0}^{k-1} \left(\numtrees{u}{v_i}{v_{i+1}}{v} - \numtrees{u}{v_{i+1}}{v_i}{v}\right)\,,
  \end{align}
  \endgroup
  where $\numtrees{u}{v_i}{v_{i+1}}{v}$ is the number of spanning trees in which the (unique) path from $u$ to $v$ contains $(v_i, v_{i+1})$ in that order and $N$ is the number of all spanning trees of the graph $G$.
}
The main idea of Ref.~\cite{apgm20} is to compute a sufficiently large sample of
uniform spanning trees (USTs) in order to approximate the effective resistances according to Eq.~(\ref{eq:N-values-in-path}).
The resistance values are then used for approximating the diagonal entries of $\Lpinv$, together with one
column of $\Lpinv$ derived from solving one linear system.

\subsubsection{Updating Approximate $\diag{\Lpinv}$ after Edge Insertions}
\label{sub:update-lpinv}
\extend{For updating $\diag{\Lpinv}$ within $k$-GRIP,} we need to sample USTs for every new graph $G_{r+1}$ (in round $r$).
We do so during the update step of Algorithm~\ref{alg:stochastic} (Line~\ref{line:update}) and save computations by reusing previously computed USTs corresponding to
$G_{r}$. \extend{This dynamic approximation approach can also be useful in other contexts.}
The reused trees are not uniformly distributed in the new graph
$G_{r+1} := G_{r} \cup \{a,b\}$, however, and need to be reweighted accordingly.
Moreover, we still need to sample a number of USTs corresponding
to trees of $G_{r+1}$ that contain the additional edge $\{a,b\}$.
To do so, we use a variant of Wilson's algorithm~\cite{Wilson96}.
The final sample set is the
union of the reweighted USTs (originally from $G_{r}$) and the newly
sampled USTs in $G_{r+1}$. \extend{We provide the details in the following.}

\extend{
  To account for an edge insertion into $G$, let the set of all spanning trees of $G$ (before the edge insertion)
  be denoted as $\mathcal T = \mathcal T_G$.
  When looking at the potential difference between two nodes $u$ and $v$ within one particular spanning tree $T$, then the electrical flow induced on each edge on the unique path from $s$ to $t$ in $T$ is $1/N$.
  Using the principle of superposition for the electrical flow in $G$, we can then write $\effres{u}{v} = \sum_{i=0}^{k-1} \vent{f}{v_i, v_{i+1}}  = \sum_{ T \in \mathcal T} \sum_{i=0}^{k-1} \vent{f^{\mathnormal{(T)}}}{v_i, v_{i+1}}$, where $\vent{f^{\mathnormal{(T)}}}{\cdot}$ restricts the electrical flow
  to edges of the respective spanning tree $T$ (edges not in $T$ contribute $0$ to the sum).
  In the following, we use $\myvec{F_{(u,v)}}(T) := \sum_{i=0}^{k-1} \vent{f^{\mathnormal{(T)}}}{v_i, v_{i+1}}$ as short-hand notation for the sum of the flows.
  Now let $G'$ be the new graph when an edge $e= \{u, v\}$ is added to the graph $G$.
  Let $\tau$ be a random variable from the uniform distribution over spanning trees of
  $G$. Then $\effres{u}{v}  = \EE { \myvec{F_{(u,v)}}(\tau) }$
  and we are interested in computing their updated values upon edge insertions.
}

\extend{
We define $\mathcal T'\defeq \mathcal T_{G'}$. Let $\tau'$ be a uniformly
distributed valued random variable over $\mathcal T'$. We consider
$\myvec{F^{'}_{(u,v)}}\colon \mathcal T'\to \mathbb{R}$ and denote by
$\myvec{F_{(u,v)}} = \myvec{F^{'}_{(u,v)}} |_{\mathcal T}\colon \mathcal T\to \mathbb{R}$
its restriction to spanning trees of $G$.
\begin{lemma}
  \label{lem:update-trees}
  Let $G'$ be the graph resulting from inserting $e = \{u,v\}$ into $G$. Then
  \begin{align}
    \effresG{G'}{u}{v} =  \frac{\effresG{G}{u}{v}}{1+\effresG{G}{u}{v}} \EE{\myvec{F'_{(u,v)}}(\tau') \mid e \in \tau'} + \frac{1}{1+\effresG{G}{u}{v}} \EE{\myvec{F_{(u,v)}}(\tau') \mid \tau' \in \mathcal T}\,.
  \end{align}
\end{lemma}
\begin{proof}
  Recall from above that $\effresG{G'}{u}{v} = \EE{\myvec{F^{'}_{(u,v)}}(\tau')}$.
  Also note that for any edge $e'=\{u',v'\}$, it holds that its effective resistance equals
  the probability to be part of a UST. Now $\EE{\myvec{F^{'}_{(u,v)}}(\tau')}$ can be computed by distinguishing whether $e$ is contained in $\tau'$ or not:
  \begin{small}
    \begingroup
    \allowdisplaybreaks
    \begin{align}
      \begin{split}
        \EE{\myvec{F^{'}_{(u,v)}}(\tau')}
        & = \PP{ e \in \tau'} \EE{\myvec{F'_{(u,v)}}(\tau') \mid e \in \tau'} + \PP{e \notin \tau'} \EE{\myvec{F^{'}_{(u,v)}}(\tau') \mid e \notin \tau'} \\
        & = \PP{ e \in \tau'} \EE{\myvec{F'_{(u,v)}}(\tau') \mid e \in \tau'} + \PP{e \notin \tau'}\EE{\myvec{F_{(u,v)}}(\tau') \mid \tau' \in \mathcal T} \\
        & = \effresG{G'}{u}{v} \EE{\myvec{F'_{(u,v)}}(\tau') \mid e \in \tau'} + (1-\effresG{G'}{u}{v}) \EE{\myvec{F_{(u,v)}}(\tau') \mid \tau' \in \mathcal T} \\
        & = \frac{\effresG{G}{u}{v}}{1+\effresG{G}{u}{v}} \EE{\myvec{F'_{(u,v)}}(\tau') \mid e \in \tau'} + \frac{1}{1+\effresG{G}{u}{v}} \EE{\myvec{F_{(u,v)}}(\tau') \mid \tau' \in \mathcal T}\,,
        \label{eq:update-trees}
      \end{split}
    \end{align}
    \endgroup
  \end{small}
  using $\PP{ e \in \tau'} = \effresG{G'}{u}{v} = \frac{\effresG{G}{u}{v}}{1+\effresG{G}{u}{v}}$
  (the latter equation follows from Ref.~\cite[Cor.~3]{DBLP:conf/cocoa/RanjanZB14} by setting $u=x=i$ and $v=y=j$)
  and the fact that $\mathcal T$ equals $\mathcal T'\setminus \mathcal T_e$, where $\mathcal T_e$ is the set of trees containing $e$.
\end{proof}
}

\paragraph{Adapting the UST Algorithm}

\extend{
  The second term in Eq.~(\ref{eq:update-trees}) can be approximated using the USTs of $G$, which are already available from previous rounds of the algorithm.
  To approximate the first term, one can sample spanning trees of $G'$ which contain $e$.
  For this we use Algorithm~\ref{alg:wilson}, which is a slight adaptation of Wilson's algorithm
  with a modified starting state. A spanning tree which contains $\{u,v\}$
  can be reinterpreted as a forest with two components by removing $\{u,v\}$.
  Thus, we initialize our version of Wilson’s algorithm with a forest $T$ with two components where each component
  contains only one of $u$ and $v$. Then in each
  iteration we generate a loop-erased random walk from a random vertex until it hits a node
  in $T$.
}

\extend{
  \begin{proposition}
    The distribution of forests $T$ sampled by Algorithm \ref{alg:wilson} is the uniform distribution on the
    set of all spanning trees which contain the edge $\{u, v\}$.
  \end{proposition}
  \begin{proof}
    Avena \etal~\cite{avena2018random} reformulate Wilson's algorithm for uniform spanning forests (USFs) and multiple roots (one for each tree in the forest).
    That is why we set $u$ and $v$ as the roots of two separate trees and let the algorithm compute a USF with two trees. The two trees in the USF are then linked by the edge $\{u,v\}$, resulting in a spanning tree $T'$ of $G'$. By the USF property of the two trees above the claim follows.

  \end{proof}
}

\begin{algorithm}[H]
  \begin{algorithmic}[1]
    \begin{small}
      \Function{Sampling}{$G$, $a$, $b$}
      \State \textbf{Input:} Graph $G=(V,E)$, edge $\{a, b\} \in E$
      \State \textbf{Output:} $T$: UST of $G$ containing $\{a, b\}$
      \State $T_1 \gets$ tree consisting of $a$
      \State $T_2 \gets$ tree consisting of $b$
      \State  Let $x_1, \ldots, x_{n-2}$ be an arbitrary ordering of $V\setminus \{a,b\}$
      \For {$i \gets 1, \dots, n-2$}
      \State $ P \gets$ a random walk from $x_i$ to either $T_1$ or $T_2$ \label{line:candidate}
      \State add the loop erasure of $P$ to the tree hit by $P$
      \EndFor
      \State \textbf{return} $T_1 \cup T_2 \cup \{a,b\}$
      \EndFunction
    \end{small}
  \end{algorithmic}
  \caption{Algorithm for sampling a UST of $G$ containing a fixed edge $\{a,b\}$}
  \label{alg:wilson}
\end{algorithm}

\paragraph{Putting the Pieces Together}

\extend{
  \begin{algorithm}[H]
    \begin{algorithmic}[1]
      \begin{small}
        \Function{ApproxUpdateDiag}{$G_r, r, u, U, t, w, R, B_u$}
        \State \textbf{Input:} Graph $G_r = G \cup \{a, b\}$, current round $r$, pivot node $u$, UST container $U[]$,
        total \# of USTs $t$, round weights $w[]$, effective resistance estimates $R[]$, BFS Tree $B_u$
        \State \textbf{Output:} $\operatorname{diag}(\widetilde{\Lpinv_{G'}})$
        \State $R_{new}[v] \gets 0 ~\forall v \in V\setminus \{u\}$ \label{line:start-init}
        \State $\omega \gets \effresG{G'}{a}{b} = \frac{\effresG{G}{a}{b}}{1 + \effresG{G}{a}{b}}$  \Comment{computed via $\Lpinv_{G'}[:, a]$ and $\Lpinv_{G'}[:, b]$ (linear systems)} \label{line:effres-upd}
        \For{ $i = 0, \dots, r-1$ }
        \State $w[i] \gets (w[i] \cdot (1-\omega))$ \label{line:weight-upd}
        \State $U[i]$.resize($\ceil{w[i] \cdot t}$)  \Comment {\small adjust \# of USTs for round $i$ acc.\ to round weights} \label{line:repo-upd}
        \EndFor
        \State $w$.append($\omega$) \Comment{add weight of current round} \label{line:weight-ext}
        \For{$i \gets 1$ to $\ceil{\omega \cdot t}$}\label{line:ust-sampling-loop}
        \Comment {\small $\ceil{\omega \cdot t}$ times}
        \State $T_i \gets $ \textsc{Sampling}($G_r$, $a$, $b$) \label{line:sampling}
        \Comment{\small $\Oh(m \diam{G})$}\label{line:ust-sampling}
        \State $R_{new} \gets$ \textsc{Aggregate}($T_i$, $R_{new}$, $B_u$) \label{line:aggregation}
        \Comment {\small $\Oh(n \diam{G})$}
        \State $U[r].$append($T_i$) \label{line:store-ust}
        \EndFor \label{line:end-sampling}
        \State $R_{new} \gets R_{new}/\ceil{\omega \cdot t}$ \label{line:effres-scale}
        \State $R \gets \omega R_{new} + (1-\omega) R$ \Comment{Acc.\ to Lemma~\ref{lem:update-trees}} \label{line:new-R}
        \For{$v \in V\setminus \{u\}$}
        \Comment {\small All iterations: $\Oh(n)$}
        \State $\ment{\widetilde{\mat{L}^\dagger_{G'}}}{v}{v} \gets R[v] - \ment{\widetilde{\mat{L}^\dagger_{G}}}{u}{u} + 2 \ment{\widetilde{\mat{L}^\dagger_{G}}}{v}{u}$ \label{line:final-aggregation}
        \EndFor \label{line:end-fill}
        \State \textbf{return} $\operatorname{diag}(\widetilde{\Lpinv_{G'}})$
        \EndFunction
      \end{small}
    \end{algorithmic}
    \caption{Compute $\diag{\Lpinv_{G_r}}$ upon edge insertion}
    \label{alg_lpinv_ust_edge_added}
  \end{algorithm}
}

\extend{
  By applying Eq.~(\ref{eq:update-trees}) to the effective resistance estimates, we obtain
  Algorithm~\ref{alg_lpinv_ust_edge_added}. It obtains an approximation for $\diag{\Lpinv_{G'}}$,
  where $G'$ is obtained from $G$ by inserting an edge $e=\{a,b\}$.
  This algorithm is run each time after an edge is added to $G$.
  To obtain an initial set of USTs, the algorithm of Angriman \etal~\cite{apgm20} is applied to the original graph $G$.
  These USTs are stored in what we call the UST repository, which is used to also store USTs
  from graphs resulting from a series of edge insertions. All spanning trees together in this repository form
  a sufficiently large sample of USTs for the graph of the current round.
  Lines~\ref{line:start-init} and~\ref{line:effres-upd} initialize the vector of new resistance estimates and
  compute the effective resistance $\omega$ of the inserted edge $\{a, b\}$. The latter is necessary to scale the contribution
  of the USTs from this and previous rounds according to Lemma~\ref{lem:update-trees} (Line~\ref{line:new-R}).
  How many USTs each round contributes is governed by the round weight $w$; both numbers have to be adapted according to $\omega$
  (Lines~\ref{line:weight-upd} and~\ref{line:repo-upd}).
  After sampling and aggregating the new trees as well as updating $R$ (Lines~\ref{line:ust-sampling-loop} to~\ref{line:new-R}), the new diagonal
  approximation can be computed and returned.
}

\subsection{ {\sc *StochJLT}}
In this section we propose an improvement \extend{to \textsc{SimplStoch}} that exploits the following
observation: to evaluate the gain function for an arbitrary vertex pair
$\{a,b\}$, we only require to compute
the squared $\ell_2$-norm of two distance vectors:
$\shdG{G}{a}{b} = \norm{\Lpinv(\uvec{a} - \uvec{b})}^2$ and
$\effresG{G}{a}{b} = \norm{\mat{B}^T \Lpinv (\uvec{a} - \uvec{b})}^2$ (Eq.~(\ref{eq:denominator}-\ref{eq:nominator})).
Viewing $\shdG{G}{a}{b}$ and $\effresG{G}{a}{b}$ as pair-wise distances
between vectors in $\{\Lpinv\}_{a \in V}$ and
$\{\mat{B}^T\Lpinv\}_{a \in V}$ (respectively) allows us to apply the
Johnson-Lindenstrauss transform (JLT)~\cite{Johnson1984ExtensionsOL}.
In this case, pairwise distances among vectors are
nearly preserved if we project the vectors
onto a low-dimensional subspace,
spanned by \extend{$\Oh(\log{n}/\eta^2)$} random vectors.
The JLT lemma, \extend{in the improved version by Dasgupta and Gupta~\cite{DBLP:journals/rsa/DasguptaG03},} can be stated as:
\begin{lemma}
  \label{jlt}
  Given fixed vectors $\myvec{u_1} \ldots, \myvec{u_n} \in \mathbb{R}^d$ and $\eta > 0 $, let $\mat{Q} \in \mathbb{R}^{q\times d}$ be \extend{a random Gaussian matrix with entries from N(0,1) and $q > 24 \log{n}/\eta^2$.} Then with probability at least $1-1/n$
  \begin{equation}
    (1-\eta) \norm{\myvec{u_i}-\myvec{u_j}}^2 \leq \norm{\mat{Q}\myvec{u_i} - \mat{Q}\myvec{u_j}}^2 \leq (1+\eta)\norm{\myvec{u_i}-\myvec{v_j}}^2 \label{eqn_jlt_ineq}
  \end{equation}
  for all pairs $i,j \leq n$.
\end{lemma}

Using Lemma \ref{jlt}, we can simply project matrices
$\Lpinv$  
and $\mat{B} \Lpinv$ 
onto $q$ vectors, \ie the $q$ rows of some random
matrices $\mat{P} \in \mathbb{R}^{q\times n}$ and $\mat{Q} \in \mathbb{R}^{q\times m}$, respectively.
To actually reduce the overall computation time,
we need to avoid the involved pseudoinversion. For that, we resort to
efficient linear system solvers.
Thus, combining the random projections technique with fast linear
solvers, one can approximate distances between vertex pairs
within a factor of $(1 \pm \eta)$
in $\Oh(I(n,m)\log{n} /\eta^2)$ time, where $I(n, m)$
is the running time of the Laplacian solver.

Hence to approximate $\shdG{G}{a}{b}$ and $\effresG{G}{a}{b}$,
we compute the projected distances
$\norm{\mat{P}\Lpinv (\uvec{a} - \uvec{b})} ^2$ and $\norm{\mat{Q}\mat{B}\Lpinv (\uvec{a} - \uvec{b})} ^2$, respectively.
\extend{One can} avoid the solution of two sets of Laplacian systems
by expressing the effective resistances directly
via the projection of (squared) biharmonic distances onto the lower dimensional space.
More precisely, \extend{one only solves} $\Lap \mat{Y} = \mat{P}^T-\frac{1}{n} \1\1^T \mat{P}^T$.
Due to $\Lpinv \cdot \frac{1}{n} \1\1^T = \mat{O}$ (the zero matrix), it follows $\mat{Y} = \Lpinv \mat{P}^T$, so that we can express effective resistances as follows:
\begin{equation}
  \begin{small} {
      \label{eq:jlt}
      \begin{split}
        & \norm{\mat{Q}\mat{B} \mat{Y} \mat{P} (\uvec{a} - \uvec{b})}^2 =
        (\uvec{a} - \uvec{b})^T \mat{P}^T \mat{Y}^T \mat{B}^T \mat{Q}^T\mat{Q} \mat{B} \mat{Y} \mat{P} (\uvec{a} - \uvec{b}) \\
        & = (\uvec{a} - \uvec{b})^T \Lpinv \mat{B}^T \mat{B} \Lpinv (\uvec{a} - \uvec{b}) =
        \norm{\mat{B} \Lpinv (\uvec{a} - \uvec{b})}^2,
      \end{split}
    }
  \end{small}
\end{equation}
\extend{where we assume that $\mat{Q}$ and $\mat{P}$ are orthonormal matrices. Note that there are formulations of the JLT with orthonormal matrices, including very early ones~\cite{Johnson1984ExtensionsOL,FRANKL1988355}.
  The formulation in Lemma~\ref{jlt} with random Gaussian entries is only ``almost'' orthogonal; this condition is usually sufficient in practice~\cite{DBLP:journals/jcss/Achlioptas03}.
  In our case this would mean that the equality in Eq.~(\ref{eq:jlt}) becomes ``approximately equal'', which would be sufficient for our heuristics as well.}

We can integrate the JLT approximation both in the context
of {\sc ColStoch} and {\sc SimplStoch} (having {\sc ColStochJLT} and {\sc SimplStochJLT}, respectively).
\extend{For both approaches, we set $\eta := 0.55$ in our experiments and thus consider it as a constant in the time complexity statements regarding \textsc{*StochJLT}.}
Let us consider the case of {\sc ColStochJLT}:
Again, the compute step is performed after selecting
set $\mathcal S$ (just after Line~\ref{line:candidate}).
Indeed, we compute the vectors in $\{\Lpinv\}_{a \in \mathcal S}$ and
$\{\mat{B}\Lpinv\}_{a \in \mathcal S}$ for $G_0$, where
$s := |\mathcal{S}| = n \sqrt {\frac 1 {k} \cdot\log(\frac 1 {\delta})}$.
Since, later, we only perform evaluations for pairs in
$\mathcal{S} \times \mathcal{S}$, it suffices to consider projections
onto $\log s$ rows (via $\mat{P} \in \mathbb{R}^{\log s \times n}$ and $\mat{Q} \in \mathbb{R}^{\log s \times m}$).
During the main loop of Algorithm~\ref{alg:stochastic} we perform the same number of overall function evaluations as in {\sc ColStoch}, that is $\Oh'(n^2)$. However, now a single function evaluation for an
arbitrary vertex pair takes $\mathcal O(\log s)$ via the formula
\begin{equation}
  \label{eq:gain-approx-jlt}
  \gain{a}{b} \approx \frac{\norm{\mat{P}\Lpinv(\uvec{a} - \uvec{b}))^2}} {1 + \norm{\mat{Q}\mat{B} \mat{Y} \mat{P} (\uvec{a} - \uvec{b})^2}}
\end{equation}
(up to a relative error of $(1+\eta)$).
For the update step, we need to sample new projections $\mat{P}$ and
$\mat{Q}$ and recompute the two matrices $\mat{P}\Lpinv$ and $\mat{Q}\mat{B} \mat{Y} \mat{P}$.
\extend{
  The dominant cost of the approach is due to the main loop,
  which takes $\mathcal O'( n^2 \log s)$ time.
  For {\sc SimplStochJLT}, the time
  complexity is $\mathcal O'( n^2 \log n)$.
}

\subsection{\textsc{SpecStoch}}
\label{sub:specstoch}
As the last approach in this section we propose to exploit the spectral expression of the gain function. More precisely, we combine the spectral expressions  of effective resistance and (squared) biharmonic distance~(Eqs.~\eqref{eq:nominator}
and~\eqref{eq:denominator}) to write Eq.~(\ref{eq:gain}) as
\begin{equation}
  \label{eq:gain-eigen}
  \begin{split}
    \gain{a}{b}  = n\cdot \frac{\sum_{i = 2}^{n} \frac{1}{(\lambda_i)^2} \cdot (\vent{u_i}{a} - \vent{u_i}{b})^2} {1 + \sum_{i = 2}^{n} \frac{1}{\lambda_i} \cdot (\vent{u_i}{a} - \vent{u_i}{b})^2} \,.
  \end{split}
\end{equation}

Eq.~(\ref{eq:gain-eigen}) benefits from the fact that both
effective resistance and (squared) biharmonic distance only depend on the spectrum
of the same matrix $\Lap$. Still, the full spectral decomposition
of $\Lap$ incurs $\Oh(n^3)$ time and is equally prohibitive as
computing $\Lpinv$ for larger $G$. To reduce the complexity,
we propose an approximation of Eq.~(\ref{eq:gain-eigen}) using
standard low-rank techniques~\cite{bozzo2012effective} and new bounds for both distances.
To do so, we exploit the fact that the bulk of the eigenvalues
tends to concentrate away from the smallest eigenvalues~\cite{Chung06}.
Moreover, we compute only a small number of eigenpairs
on the lower side of the spectrum of $\mat{L}$.
We expect that the smaller eigenpairs have a larger influence
on the sums of Eq.~(\ref{eq:gain-eigen}):
for small $i$, contributions 
are accentuated by a large weight, $\frac{1}{\lambda_i^2}$ (recall that we index the eigenvalues
ordered non-decreasingly).
At the same time, the entries of eigenvector $\myvec{u_i}$ fluctuate slowly,
so we should carefully select $\{a,b\}$ to avoid near-zero contributions.
On the other hand, for large $i$, the eigenvectors $\myvec{u_i}$
fluctuate rapidly, since they correspond to high
frequency modes of the spectrum~\cite{spielman2012spectral}.
Their contribution to Eq.~(\ref{eq:gain-eigen})
is undermined by $\frac{1}{\lambda_i^2}$ (small for large $i$).
The above observations suggest that for a new edge insertion
$\{a,b\}$, the focus should be on eigenpairs corresponding to small $i$.

We now show how to derive bounds for $\shdG{G}{a}{b}$.
First we break Eq.~(\ref{eq:nominator}) into partial sums where $c \leq n$ is a cut-off value.

\begin{equation}
  \label{eq:up-bound}
  { \small
    \begin{split}
      & \shdG{G}{a}{b} 
      =  \sum_{i = 2}^{c} \frac{(\vent{u_i}{a} - \vent{u_i}{b})^2} {\lambda_i^2}  + \sum_{i = c+1}^{n} \frac{(\vent{u_i}{a} - \vent{u_i}{b})^2}{{\lambda_i}^2}\\
      & \leq \sum_{i = 2}^{c} \frac{(\vent{u_i}{a} - \vent{u_i}{b})^2}{{\lambda_i}^2}   + \frac{1}{{\lambda_c}^2} \sum_{i = c+1}^{n} (\vent{u_i}{a} - \vent{u_i}{b})^2 \\
      & \leq \sum_{i = 2}^{c} \frac{(\vent{u_i}{a} - \vent{u_i}{b})^2}{{\lambda_i}^2}   + \frac{1}{{\lambda_c}^2} (2 - \sum_{i = 2}^{c}(\vent{u_i}{a} - \vent{u_i}{b})^2) \\
      & = \frac{2}{{\lambda_c}^2} + \sum_{i = 2}^{c} (\frac{1}{{\lambda_i}^2} - \frac{1}{{\lambda_c}^2}) (\vent{u_i}{a} - \vent{u_i}{b})^2   \,.
    \end{split}
  }
\end{equation}
The first inequality holds for large enough eigenvalues ($ \geq 1$),
since $\lambda_c \leq \lambda_{c+i}$ and $\frac{1}{(\lambda_c)^2} \geq \frac{1}{(\lambda_{c+i})^2}$  for any $i$.
Moreover, the third line comes from the following observation:
\begin{equation}
  \label{eq:low-bound}
  \begin{split}
    & \sum_{i = 2}^{n} (\vent{u_i}{a} - \vent{u_i}{b})^2 = \sum_{i = 1}^{n} (\vent{u_i}{a} - \vent{u_i}{b})^2 \\
    & = \sum_{i = 1}^{n} {\vent{u_i}{a}}^2 + \sum_{i = 1}^{n} {\vent{u_i}{b}}^2 - 2\sum_{i = 1}^{n} \vent{u_i}{a} \vent{u_i}{b} \\
    & = \norm{\myvec{u^T_a}}^2 + \norm{\myvec{u^T_b}}^2 -2 \mat{U}[a,:] \mat{U^T}[:,b] = 2
  \end{split}
\end{equation}
for $a \neq b$ since $\mat{U}$ is double-orthogonal. Moreover:
\begin{equation}
  \label{eq:low-bound}
  { \small
    \begin{split}
      & \shdG{G}{a}{b} 
      =  \sum_{i = 2}^{c} \frac{(\vent{u_i}{a} - \vent{u_i}{b})^2} {\lambda_i^2}  + \sum_{i = c+1}^{n} \frac{(\vent{u_i}{a} - \vent{u_i}{b})^2}{{\lambda_i}^2}\\
      & \geq \sum_{i = 2}^{c} \frac{(\vent{u_i}{a} - \vent{u_i}{b})^2}{{\lambda_i}^2}   + \frac{1}{{\lambda_n}^2} \sum_{i = c+1}^{n} (\vent{u_i}{a} - \vent{u_i}{b})^2 \\
      & \geq \sum_{i = 2}^{c} \frac{(\vent{u_i}{a} - \vent{u_i}{b})^2}{{\lambda_i}^2}   + \frac{1}{{\lambda_n}^2} (2 - \sum_{i = 2}^{c}(\vent{u_i}{a} - \vent{u_i}{b})^2) \\
      & = \frac{2}{{\lambda_n}^2} + \sum_{i = 2}^{c} (\frac{1}{{\lambda_i}^2} - \frac{1}{{\lambda_n}^2}) (\vent{u_i}{a} - \vent{u_i}{b})^2           \,,
    \end{split}
  }
\end{equation}
where the inequality in the third line holds, since $\lambda_n \geq \lambda_{c+i}$ for any $i$.
Following the above, we can easily derive similar bounds for $\effresG{G}{a}{b}$.
Plugging those bounds together, we can approximate Eq.~(\ref{eq:gain-eigen}) using the following inequality:
\begin{equation}
  \label{eq:gain-low-approx}
  { \small
    \begin{split}
      & \frac{ \frac{2}{{\lambda_c}^2} + \sum_{i = 2}^{c} (\frac{1}{{\lambda_i}^2} - \frac{1}{{\lambda_c}^2}) (\vent{u_i}{a} - \vent{u_i}{b})^2} { 1 + \frac{2}{\lambda_n} + \sum_{i = 2}^{c} (\frac{1}{\lambda_i} - \frac{1}{\lambda_n}) (\vent{u_i}{a} - \vent{u_i}{b})^2}
      \leq   \gain{a}{b} \\
      & \leq \frac{ \frac{2}{{\lambda_n}^2} + \sum_{i = 2}^{c} (\frac{1}{{\lambda_i}^2} - \frac{1}{{\lambda_n}^2}) (\vent{u_i}{a} - \vent{u_i}{b})^2} { 1 + \frac{2}{\lambda_c} + \sum_{i = 2}^{c} (\frac{1}{\lambda_i} - \frac{1}{\lambda_c}) (\vent{u_i}{a} - \vent{u_i}{b})^2} \,.
    \end{split}
  }
\end{equation}

Adapting the general framework of Algorithm~\ref{alg:stochastic}
for {\sc SpecStoch} is rather straightforward:
In Line~\ref{line:compute} we compute the first $c$
eigenpairs along with the largest eigenvalue of $\Lap$ (corresponding to $G_0$).
We do so using standard iterative methods,
such as the Lanczos algorithm~\cite{PAIGE80},
which often takes only $\Oh(c m)$ time for sparse matrices~\cite{koch20118},
depending on the desired accuracy and eigenvalue distribution.
During the main loop, the algorithm
performs $\Oh'(n^2)$ function evaluations (dictated by the stochastic approach).
\extend{Assuming ``well-behaved'' eigenvalues,} single function evaluations in Line~\ref{line:single-eval} require only
$\Oh(c)$ time using the bounds in Eq.~(\ref{eq:gain-low-approx}).
Finally, we update the eigenpairs of $G_{r+1}$ in Line~\ref{line:update}.
To speed up the update step, we bootstrap the solution of the eigensolver
with the solution of the previous round.
\extend{Under our assumptions,
  the overall complexity of {\sc SpecStoch} is $\Oh'(n^2 c + k c m)$
  and in case both $c \in \Oh(1)$ and $k \in \Oh(1)$, the overall time complexity becomes $\Oh'(n^2)$.
}

\section{Heuristics for $k$-LRIP}
\label{sec:lrip}
\extend{
  Recall the idea of the $k$-LRIP problem:
  consider a fixed \emph{focus node} $v$.
  How can the robustness of the graph be improved when we restrict the edges that may be added to the graph to those that are incident to $v$?
  This problem is a local variant of $k$-GRIP in the sense that we can only add edges local to $v$.
  Still, we take a global view of the graph and try to improve the total graph resistance with no special consideration for $v$.

  Now assume there is a set $F$ of focus nodes and for each $v\in F$ we want to solve the $k$-LRIP problem independently.
  Then it makes sense to run the preprocessing steps of our algorithms just once and re-use the results when solving $k$-LRIP for each $v\in F$.

  In the following subsections we will describe how we adapt the heuristics from Section~\ref{sec:contrib} to $k$-LRIP.
  Let us mention a few general aspects first.
  Since we still optimize for the total graph resistance, the formulas derived for $k$-GRIP can generally be re-used; the gain only becomes a function of one (fixed focus) node now.
  Also the basic structure of Algorithm~\ref{alg:stochastic} remains the same in general.
  Some changes to note: recall from Section~\ref{sec:prelim} that the set of all candidates is $\Omega_v$.
  A candidate edge $e=\{v,b\}$ from this set is uniquely identified by $b$.
  That is why $a$ equals $v$ in Lines~\ref{line:candidate-iteration}-\ref{line:graph-update}.
  Moreover, in Line~\ref{line:candidate-iteration} we sample from $\mathcal{S}$ instead of $\mathcal{S} \times \mathcal{S}$.

  Compared to $k$-GRIP, the candidate set for $k$-LRIP is considerably smaller.
  This reduces the number of evaluations in each iteration of the main loop (per focus node).
  Table~\ref{tbl:complexities-lrip} shows the time and space complexities of all approaches for $k$-LRIP.
  For some heuristics, depending on the density of the graph, the dominant term becomes either the total number of evaluations (second column) or the update step (third column). If $m$ is considerably larger than $n$, \textsc{SimplStoch} may actually provide the best overall time complexity.

  \begin{table}[H]
    \caption{\extend{Time complexities (assuming standard (pseudo)inversion tools, linear solvers, and eigensolvers used in practice for Laplacians of general graphs)
        of all approaches involved for $k$-LRIP for one focus node. Columns correspond to major steps of Algorithm~\ref{alg:stochastic}.
        The $\widetilde{\ensuremath{\mathcal{O}}}$-notation hides $\log(1/\epsilon)$ factors, where $\epsilon$
        is the accuracy threshold of the linear solver. The ${\mathcal{O}}'$-notation
        hides $\log(1/\delta)$ factors, where $\delta$ determines the sample size in the stochastic candidate selection.
        Note that we consider the JLT parameter $\eta$ as a constant. The time complexity of the compute step is amortized over all focus nodes $F$.
        More details in the text.} }
    \centering

    \resizebox{\columnwidth}{!}{%
      \begin{tabular}{lllll} \toprule
                            & Compute                                                            & \#Evals $\times$ SingleEval      & \extend{All updates}                                    & Memory                       \\
        \midrule
        {\sc StGreedy}      & $\Oh(\frac{n^3}{|F|})$                                             & $\Oh(kn)$ $\times$ $\Oh(n)$      & $\Oh(kn^2)$                                             & $\Oh(n^2)$                   \\
        {\sc SimplStoch}    & $\Oh(\frac{n^3}{|F|})$                                             & $\Oh'(n)$ $\times$ $\Oh(n)$      & $\Oh(kn^2)$                                             & $\Oh(n^2)$                   \\
        {\sc ColStoch}      & $ \widetilde{\ensuremath{\mathcal{O}}}(\frac{sm\log n}{|F|} )$     & $\Oh'(n)$ $\times$ $\Oh(n)$      & $ \widetilde{\ensuremath{\mathcal{O}}}(ksm\log n )$     & $\Oh((s+\tau)n) + m)$        \\
        {\sc SimplStochJLT} & $\widetilde{\ensuremath{\mathcal{O}}}(\frac{m\log n }{|F|})$       & $\Oh'(n)$ $\times$ $\Oh(\log n)$ & $\widetilde{\ensuremath{\mathcal{O}}}(km\log n )$       & $\Oh((s+\log n)n + m)$       \\
        {\sc ColStochJLT}   & $\widetilde{\ensuremath{\mathcal{O}}}(\frac{m\log n \log s}{|F|})$ & $\Oh'(n)$ $\times$ $\Oh(\log s)$ & $\widetilde{\ensuremath{\mathcal{O}}}(km\log n \log s)$ & $\Oh((\log s + \tau) n + m)$ \\
        {\sc SpecStoch}     & $\Oh(\frac{c m}{|F|})$                                             & $\Oh'(n)$ $\times$ $\Oh(c)$      & $\Oh(k c m)$                                            & $\Oh(c n + m)$               \\
        \bottomrule\end{tabular}
    }
    \label{tbl:complexities-lrip}
  \end{table}
}

\subsection{{\sc SimplStoch}}
\extend{
  In the case of \textsc{SimplStoch}, preprocessing includes the computation of the full pseudoinverse.
  When solving $k$-LRIP for multiple focus nodes, we store a copy of the pseudoinverse before we start the main loop of Algorithm~\ref{alg:stochastic}.
  This copy is used to skip the computation of $\Lpinv$ for the other focus nodes, reducing the time complexity of the \textsc{Compute} step to $\Oh(\frac{n^3}{|F|})$ per focus node (when amortized over all focus nodes). This approach is also applied to \textsc{StGreedy}.

  Regarding sampling, we still want to inspect a subset $\mathcal S$ of $\Omega_v$.
  During each iteration of the main loop we now sample uniformly at random $s \defeq \frac{n-1 - \deg(v)}{k}\log(\frac{1}{\delta})$ vertices
  (Line~\ref{line:candidate} in Algorithm~\ref{alg:stochastic}), resulting in $\Oh'(n)$ function evaluations overall; they are performed (as in $k$-GRIP) via Eq.~(\ref{eq:lpinv_update}) applied to $\Lpinv$ (obtained during preprocessing). When an edge is added to the graph, $\Lpinv$ is updated in the same way.
}

\subsection{{\sc ColStoch}}
\extend{
  For \textsc{ColStoch}, $\mathcal S$ is sampled from $\Omega_v$ as well.
  The sample size is $s \defeq \frac{n-1-\deg(v)}{k}\log(\frac{1}{\delta})$,
  which is also the same size as $\mathcal S$ in the case of \textsc{SimplStoch}.
  The concept of sampling only specific vertices (and thus reducing the required number of columns of  $\Lpinv$) that we described for $k$-GRIP has no significance here, since all edges already have one incident node (and therefore column of $\Lpinv$) in common.
  Hence, the sets from which we sample for \textsc{SimplStoch} and \textsc{ColStoch} from $k$-GRIP are
  the same when considering a fixed focus node $v$.
  The remaining difference is that we are still using graph-based sampling probabilities as described in Section~\ref{para:strategy} (instead of uniform sampling as in \textsc{SimplStoch}) and do not compute the full pseudoinverse; instead, we solve linear systems for each column of $\Lpinv$ corresponding to $\mathcal S$ again, including the lazy update strategy described for $k$-GRIP.

  The preprocessing in \textsc{ColStoch} consists of (i) setting up a linear solver that computes the required columns of $\Lpinv$ and (ii) computing the initial sampling probabilities for $\mathcal S$, which involves approximating $\diag\Lpinv$.
  The initial states of both the solver and $\diag\Lpinv$ are stored as a copy and can then be used to setup \textsc{ColStoch} before the main loop instead of re-computing them.

  Regarding the running time of the main loop, \textsc{ColStoch} may be slower than \textsc{SimplStoch} due to the additional time for approximating $\diag\Lpinv$.
  The overall time needs to consider the preprocessing as well -- how costly that is
  with the different methods depends (mostly) on the graph size and its density.
  We expect \textsc{ColStoch} to provide higher quality results than \textsc{SimplStoch}, though, since we are using graph-based probabilities instead of uniform sampling, as discussed in Section~\ref{para:strategy}.
}

\subsection{\textsc{*StochJLT}}

\extend{
  As in the case of $k$-GRIP, we calculate $\effresG{G}{v}{b}$ and $\shdG{G}{v}{b}$ using the JLT technique.
  In \text{*StochJLT}, preprocessing involves setting up the linear solver and computing the projection with the two matrices $\mat{P}$ and $\mat{Q}$.
  Again, results can be stored and used to initialize the solver for the next focus node
  (with $G$ reset to its original state).
}

\subsection{\textsc{SpecStoch}}

\extend{

  As for $k$-GRIP, the gain function only depends on the spectrum of $\mat{L}$.
  The integration of this approach into Algorithm~\ref{alg:stochastic} is similar to $k$-GRIP:
  in the compute step, the first $c$ eigenpairs and the largest eigenpair of $\mat{L}$ are computed using iterative solvers, usually taking $\mathcal{O}(cm)$ time. These are then stored for setting up the next focus node.
  Then, in the main loop, we use the eigenpairs to compute the gain in \textsc{Eval}.
  When adding an edge to the graph, we compute the eigenpairs again (Line~\ref{line:update}) and (as before) bootstrap the new solution process with the previous round to speed up the computation.

  Since we are restricted to a fixed focus node in $k$-LRIP, the search space (and number of calls to \textsc{Eval}) is reduced when compared to $k$-GRIP. However, for \textsc{SpecStoch}, this has less of an effect on the overall running time than for the other heuristics, since in \textsc{SpecStoch} a single evaluation is rather cheap and the expensive computations are shifted to the \textsc{Compute} and \textsc{Update} steps (where we compute eigenpairs).
  Hence, we expect that \textsc{SpecStoch} performs worse for $k$-LRIP than it does for $k$-GRIP.
}

\section{Experimental Results}
\label{sec:expes}

We conduct experiments to demonstrate the performance
of our contributions compared to {\sc StGreedy}.
All algorithms are implemented in C++,
using the NetworKit~\cite{DBLP:journals/netsci/StaudtSM16} graph APIs.
Our test machine
\extend{for $k$-GRIP} is a shared-memory server
with a 2x 18-Core Intel Xeon 6154 CPU
and a total of 1.5 TB RAM.
\extend{For $k$-LRIP we use a machine with a Intel Xeon 6126 CPU and 192 GB RAM.}
To ensure reproducibility, experiments
are managed by SimexPal~\cite{angriman2019guidelines}.
Moreover, we use both synthetic and real-world input instances.
The synthetic ones follow the Erd\H{o}s-R\'enyi ({\sc ER}),
Barab\'asi-Albert ({\sc BA}) and Watts-Strogatz ({\sc WS}) models.
The real-world graphs are taken from SNAP~\cite{Leskovecxxxx} and NR~\cite{nr2015}, including
application-relevant power grid and road networks, see Table~\ref{tab:graphs}.
In this context, we consider
\extend{small graphs those whose vertex count is $< 10$K and} medium graphs those
whose vertex count is \extend{above that but below $57$K.}
The largest graph has around $129$K nodes.
To evaluate the quality of the solutions, we measure gain
improvements: $\mathcal R(G) - \mathcal R(G_k)$.
\extend{
  To this end, after selecting a new edge $\{a,b\}$ for insertion,
  $\gain{a}{b}$ is computed via a Laplacian system,
  for all approaches. This allows us to
  compare the results of different approaches in fair manner.
}
Our code and the experimental pipeline are available at~\url{https://github.com/hu-macsy/2023-kgrip-klrip}.

\extend{
  We organize our experimental evaluation in three groups:
  first, we present experiments for configuring parameters.
  Second, we evaluate all approaches for $k$-GRIP in terms of quality and running time.
  Third, we evaluate all approaches for $k$-LRIP.
}
\begin{table}[tb]
  \centering
  \caption{ \small
    Summary of graph instances, providing (in order) network name, vertex count, and edge count.}
  \label{tab:graphs}
  \scriptsize
  \begin{tabular}{|lrr|}
    \hline
    Graph                 & $|V|$ & $|E|$ \\ \hline
    inf-power             & 4K    & 6K    \\
    facebook-ego-combined & 4K    & 8.8K  \\
    web-spam              & 4K    & 37K   \\
    Wiki-Vote             & 7K    & 100K  \\
    p2p-Gnutella09        & 8K    & 2.6K  \\ \hline
    p2p-Gnutella04        & 10K   & 39K   \\
    web-indochina         & 11K   & 47K   \\
    ca-HepPh              & 11K   & 117K  \\
    web-webbase-2001      & 16K   & 25K   \\
    arxiv-astro-ph        & 17K   & 196K  \\
    as-caida20071105      & 26K   & 53K   \\
    cit-HepTh             & 27K   & 352K  \\
    ia-email-EU           & 32K   & 54.4K \\ \hline
    loc-brightkite        & 57K   & 213K  \\
    soc-Slashdot0902      & 82K   & 504K  \\
    ia-wiki-Talk          & 92K   & 360K  \\
    flickr                & 106K  & 2.31M \\
    livemocha             & 104K  & 2.19M \\
    road-usroads          & 129K  & 165K  \\
    \hline
  \end{tabular}
\end{table}

\subsection{Configuration Experiments}\label{sec:configuration-experiments}
We start by evaluating the performance of {\sc SimplStoch} for different
accuracy values on the \extend{small and} medium graphs of Table~\ref{tab:graphs}.
Following the experiments in Ref.~\cite{Mirzasoleiman2015LazierTL}, we set
the accuracy parameter $\delta$ to $0.9$ and $0.99$ (which are reasonable values according to
the experiments of Ref.~\cite{Mirzasoleiman2015LazierTL} and our own preliminary experiments).
In Table~\ref{tab:sto-accu}, we see that there
is a clear trade-off between quality and running time,
controlled by the accuracy parameter.
Still, even for a large $\delta$,
the solution of {\sc SimplStoch} is not far off compared to {\sc StGreedy},
being only 8\% off in the worst data point ($k=2$).
We also note that the solution quality
is improved as $k$ becomes larger.
To benefit from that trade-off, in the following experiments we
set $\delta$ at $0.9$ for \extend{small and} medium graphs and $0.99$
for larger ones.

\begin{table}[H]
  \caption{ 
    Quality and speedup of {\sc SimplStoch } \\ (relative to {\sc StGreedy}) for different approximation bound.
  }
  \label{tab:sto-accu}
  \centering
  \begin{tabular}{|l|c|c|c|c|c|}
    \hline
    \multirow{2}{*}{\sc{SimplStoch}} & \multicolumn{5}{c|}{Relative Quality}                                      \\
                                     & $k=2$                                 & $k=5$  & $k=20$ & $k=50$ & $k=100$ \\
    \hline
    $\delta =0.9$                    & 0.9662                                & 0.9610 & 0.9696 & 0.9810 & 0.9898  \\
    \hline
    $\delta =0.99$                   & 0.9239                                & 0.9241 & 0.9442 & 0.9559 & 0.9694  \\
    \hline
  \end{tabular}

  ~~~
  \vspace*{.15cm}

  \centering
  \begin{tabular}{|l|c|c|c|c|c|}
    \hline
    \multirow{2}{*}{\sc{SimplStoch}} &                                           %
    \multicolumn{5}{c|}{ Relative Speedup}                                       \\
                                     & $k=2$ & $k=5$ & $k=20$ & $k=50$ & $k=100$ \\
    \hline
    $\delta=0.9$                     & 2.6   & 2.6   & 2.7    & 2.7    & 3.1     \\
    \hline
    $\delta=0.99$                    & 4.0   & 3.9   & 4.2    & 4.1    & 4.6     \\
    \hline
  \end{tabular}
\end{table}

Additionally, we perform configuration experiments to determine the quality of the gain approximation
via Eq.~(\ref{eq:gain-low-approx}) for {\sc SpecStoch}.
To do so, we randomly select a vertex pair 
and compute Eq.~(\ref{eq:gain-low-approx}) for different numbers of eigenvectors.
We measure the relative error of the approximation compared
to a full spectrum computation. In Fig.~\ref{fig:eigen} we depict the results for
synthetic graphs and eigenvector number from $1$ to $n=1000$.
Even for a few tens of eigenvectors,
the relative errors for {\sc WS} and {\sc ER}
are already quite small. The relative error for {\sc BA} is larger and would require a couple of
hundreds eigenvectors to achieve a similar approximation.
\begin{figure}[tb]
  \centering
  \subfigure[\textsc{WS}: $40$ avg. degree, rewir. prob. $0.01$]{\includegraphics[clip, trim=3.5cm 8.7cm 4.5cm 8.7cm, width=0.3\textwidth]{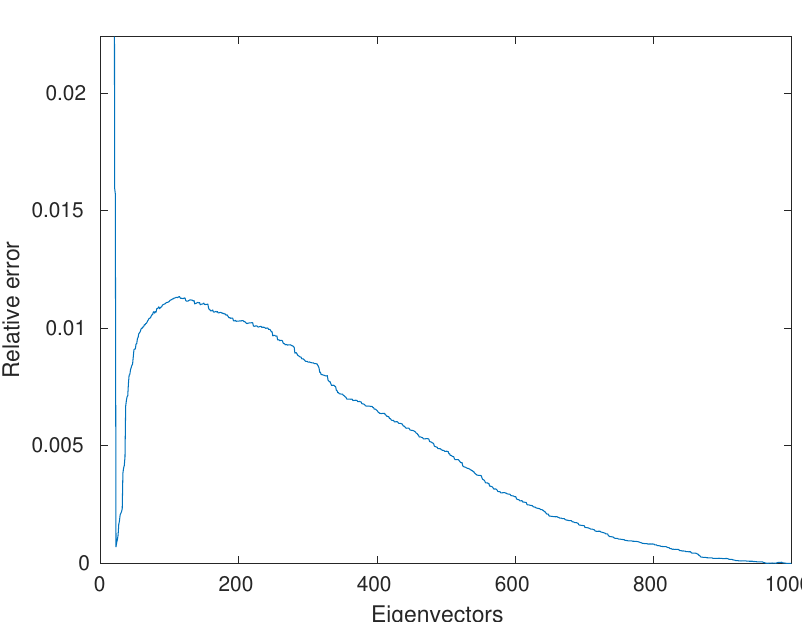}}
  \subfigure[\textsc{ER}: probability $p = 0.01$]{
    \includegraphics[clip, trim=3.7cm 8.7cm 4.5cm 8.7cm, width=0.3\textwidth]{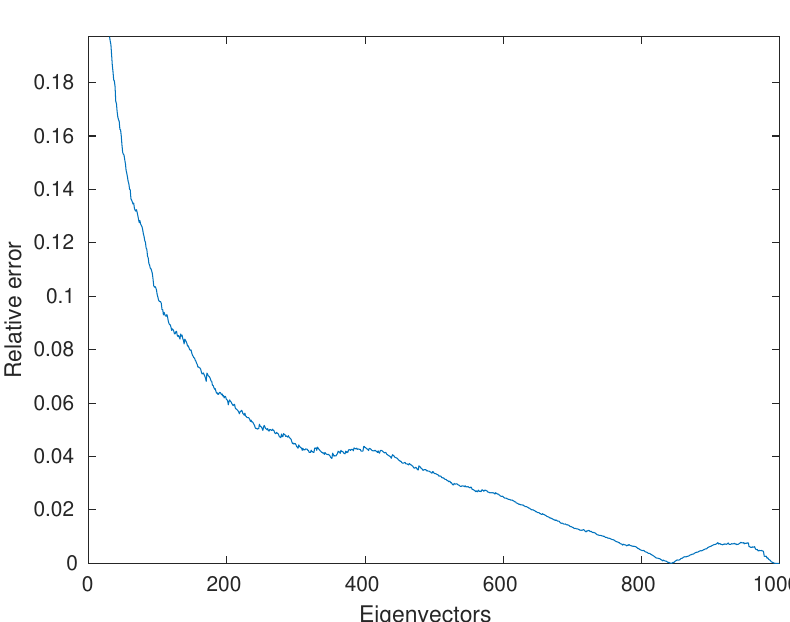}}
  \subfigure[\textsc{BA}: $m_0 = m = 4$]{\includegraphics[clip, trim=4cm 8.7cm 4.5cm 8.7cm, width=0.3\textwidth]{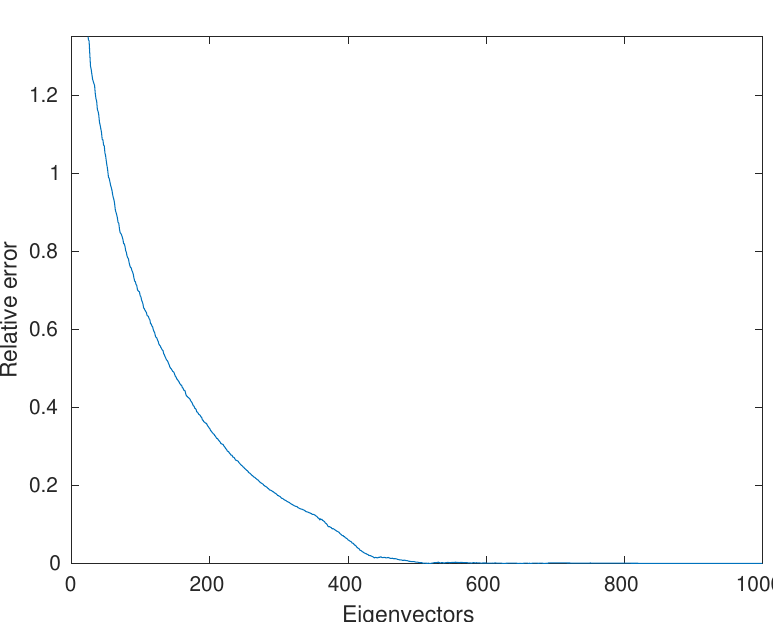}}

  \caption{Relative error of gain via Eq.~(\ref{eq:gain-low-approx}) for different number of eigenvectors.
  }
  \label{fig:eigen}
\end{figure}

Finally, we experiment with different solvers for the solution of Laplacian linear systems.
We decided to use the sparse LU solver from the Eigen~\cite{eigenweb} library
for \extend{small and} medium graphs and the LAMG solver~\cite{lamg} from NetworKit for larger ones.
We do so, because
LAMG exhibits a better empirical running time for larger complex networks than other
Laplacian solvers.
For the solution of the eigensystem (required by {\sc SpecStoch}),
we use the Slepc~\cite{Hernandez:2005} library.

\subsection{Results for $k$-GRIP}

\begin{figure}[H]
  \centering
  \subfigure[Quality]{
    \label{exp:small-quality}
    \includegraphics[width=0.46\textwidth]{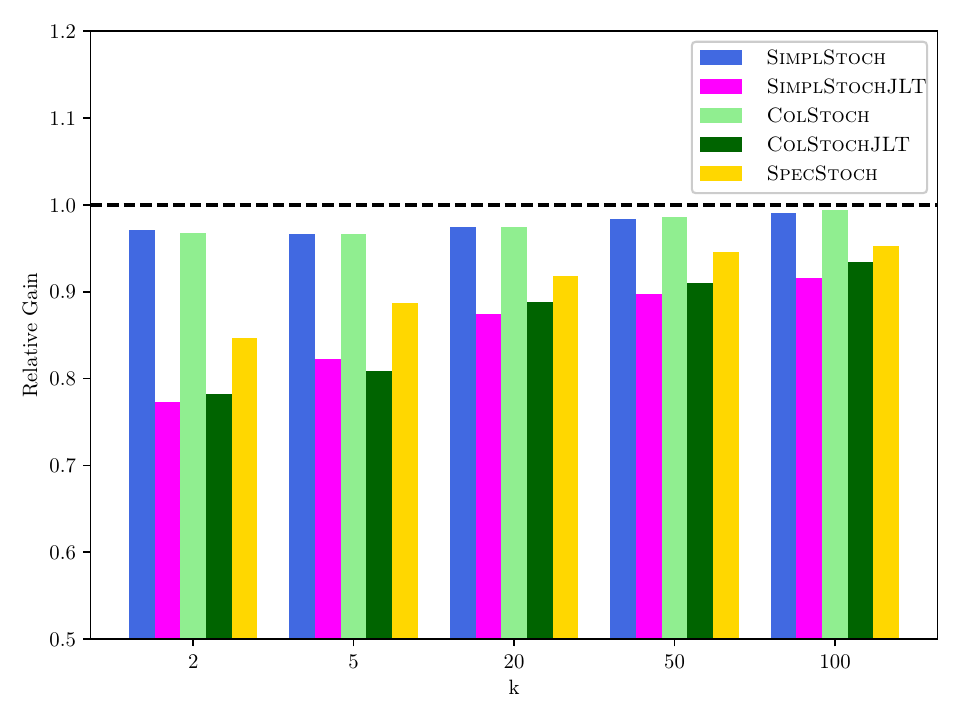}
  }
  \subfigure[Speedup]{
    \label{exp:small-time}
    \includegraphics[width=0.46\textwidth]{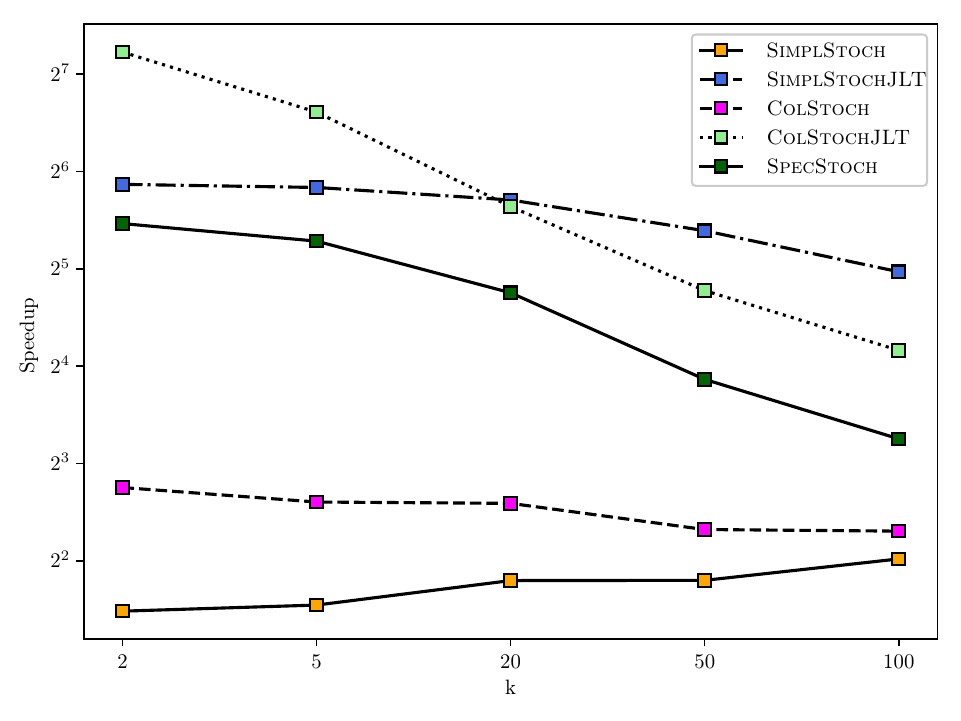}
  }
  \label{exp:comp-small}
  \caption{
    Aggregated results (via geometric mean) of $k$-GRIP on medium graphs ($n < 57K$) for different $k$. Results are relative to \textsc{StGreedy}. }
\end{figure}

We first compare our approaches on the \extend{small and} medium graphs of Table~\ref{tab:graphs}, configured
according to the previous section. Closely behind {\sc StGreedy}, {\sc SimplStoch}
and {\sc ColStoch} produce the best solutions and they
are on average 2\% away from the reference (Fig.~\ref{exp:small-quality}).
Moreover, {\sc SpecStoch}, {\sc SimplStochJLT} and {\sc ColStochJLT}
are away by 9\%, 14\% and 15\%, respectively.
Regarding running time, the
JLT-based approaches are the fastest, being on average 48$\times$ ({\sc SimplStochJLT})
and 68$\times$ ({\sc ColStochJLT}) faster than {\sc StGreedy}
(Fig.~\ref{exp:small-time}).
The scaling of {\sc ColStochJLT} is worse than that of {\sc SimplStochJLT}
for large $k$. This is due to the update step of Algorithm~\ref{alg:stochastic},
where {\sc ColStochJLT} needs to update both the effective resistance metric and the
necessary operations for JLT. Although the slowest, {\sc SimplStoch} has a good
scaling behavior as it performs only few computations in the
update step and thus is (mostly) independent of $k$.
Overall, {\sc SpecStoch} \extend{may be the
  best approach for medium graphs in a wide variety of applications since it produces good quality}
results and is on average 26$\times$ faster than {\sc StGreedy}.
\extend{Detailed runtime results are available in Table~\ref{tab:grip-runtimes}}.
A disadvantage of {\sc SpecStoch} is that the running time becomes
worse as $k$ grows due to the $k$ eigensystem updates.

\begin{table}[H]
  \centering
  \subfigure[k=2]{
    \label{tab:2grip-runtimes}
    \scriptsize
    \begin{tabular}{lrrrrrr}
      \toprule
      \multirow{2}{*}{Algorithm} & \textsc{Simpl} & \textsc{Simpl} & \textsc{Col}   & \textsc{Col}   & \textsc{Spec}  & \multirow{2}{*}{\textsc{StGreedy}} \\
                                 & \textsc{Stoch} & \textsc{JLT}   & \textsc{Stoch} & \textsc{JLT}   & \textsc{Stoch} &                                    \\
      \midrule
      inf-power                  & 50.0           & \textbf{1.6}   & 10.1           & 3.2            & 4.0            & 118.3                              \\
      facebook-ego-combined      & 18.6           & 1.7            & 5.8            & \textbf{0.7}   & 4.1            & 46.0                               \\
      web-spam                   & 29.2           & 3.4            & 17.8           & \textbf{1.8}   & 5.6            & 68.5                               \\
      Wiki-Vote                  & 110.1          & 9.0            & 65.7           & \textbf{5.4}   & 12.7           & 357.6                              \\
      p2p-Gnutella09             & 137.3          & 13.0           & 94.6           & \textbf{8.1}   & 16.3           & 296.0                              \\
      p2p-Gnutella04             & 452.5          & 38.2           & 297.0          & \textbf{28.1}  & 40.1           & 1163.2                             \\
      web-indochina              & 489.6          & 10.6           & 95.9           & \textbf{1.9}   & 15.2           & 1700.3                             \\
      ca-HepPh                   & 479.2          & 24.1           & 261.1          & \textbf{15.1}  & 31.8           & 1312.5                             \\
      web-webbase-2001           & 1634.4         & 20.1           & 292.5          & \textbf{2.4}   & 25.8           & 6402.5                             \\
      arxiv-astro-ph             & 1696.5         & 166.4          & 1426.0         & \textbf{135.4} & 165.4          & 5628.3                             \\
      as-caida20071105           & 6664.3         & 93.7           & 1434.5         & \textbf{8.0}   & 88.9           & 17544.7                            \\
      cit-HepTh                  & 4956.7         & 893.4          & 6973.9         & \textbf{816.0} & 871.2          & 13818.5                            \\
      ia-email-EU                & 11719.3        & 108.5          & 2491.7         & \textbf{5.2}   & 101.8          & 32679.4                            \\
      \bottomrule
    \end{tabular}
  }
  \subfigure[k=100]{
    \label{tab:100grip-runtimes}
    \scriptsize
    \begin{tabular}{lrrrrrr}
      \toprule
      \multirow{2}{*}{Algorithm} & \textsc{Simpl} & \textsc{Simpl} & \textsc{Col}   & \textsc{Col}   & \textsc{Spec}   & \multirow{2}{*}{\textsc{StGreedy}} \\
                                 & \textsc{Stoch} & \textsc{JLT}   & \textsc{Stoch} & \textsc{JLT}   & \textsc{Stoch}  &                                    \\
      \midrule
      inf-power                  & 41.5           & \textbf{3.4}   & 125.3          & 140.1          & 63.6            & 594.6                              \\
      facebook-ego-combined      & 19.9           & \textbf{14.3}  & 19.1           & 23.5           & 56.6            & 139.3                              \\
      web-spam                   & 31.0           & \textbf{18.3}  & 50.8           & 32.5           & 80.2            & 79.3                               \\
      Wiki-Vote                  & 121.4          & \textbf{45.1}  & 122.1          & 50.4           & 133.0           & 428.9                              \\
      p2p-Gnutella09             & 141.3          & \textbf{51.5}  & 173.7          & 54.7           & 114.5           & 314.3                              \\
      p2p-Gnutella04             & 448.2          & \textbf{129.3} & 580.5          & 133.7          & 164.6           & 1298.2                             \\
      web-indochina              & 512.6          & \textbf{18.3}  & 137.8          & 73.6           & 161.1           & 3524.8                             \\
      ca-HepPh                   & 498.4          & \textbf{88.6}  & 439.2          & 136.9          & 245.1           & 1499.1                             \\
      web-webbase-2001           & 1520.0         & \textbf{23.9}  & 295.4          & 74.5           & 209.9           & 14802.7                            \\
      arxiv-astro-ph             & 1730.6         & 469.1          & 2649.4         & 523.3          & \textbf{467.6}  & 7711.7                             \\
      as-caida20071105           & 7712.0         & 130.2          & 1630.8         & \textbf{113.1} & 475.7           & 18350.9                            \\
      cit-HepTh                  & 4932.2         & 1960.1         & 13094.1        & 2122.5         & \textbf{1544.7} & 11253.3                            \\
      ia-email-EU                & 11820.4        & 136.0          & 3000.7         & \textbf{65.7}  & 428.6           & 32771.2                            \\
      \bottomrule
    \end{tabular}
  }
  \caption{
    Runtime results in seconds for $k$-GRIP for $k=2$ and $k=100$ for medium graphs. For each instance, the fastest solver is emphasized.}
  \label{tab:grip-runtimes}

\end{table}

Finally, in Fig.~\ref{fig:large-comp} we depict results for the
large graphs of Table~\ref{tab:graphs}. For this experiment
we report absolute values since we do not have
a clear reference.
\extend{With a time limit of $12$ hours, {\sc StGreedy} always times out.}
These results show that a cubic approach such as \textsc{StGreedy}
becomes impractical once the number of nodes in the graph exceeds
a certain threshold (such as $57$K in our tests).
The best approaches for large graphs are
  {\sc ColStochJLT} and {\sc ColStoch}. Both of them produce the highest
quality results, \extend{with {\sc ColStoch} slightly ahead.}
{\sc ColStochJLT} is the fastest approach, requiring on
average $2$ [resp.\ $20$] minutes for $k=2$ [resp.\ $k = 20$].
{\sc SpecStoch} is on average as fast as {\sc ColStoch} but its performance
depends a lot on spectral properties (clustered eigenvalues or not) of
each input, as shown by the degree of skewness in Fig.~\ref{fig:large-comp}.

\begin{figure}[H]
  \centering
  \subfigure[Quality ]{
    \label{exp:large-quality}
    \includegraphics[width=0.46\textwidth]{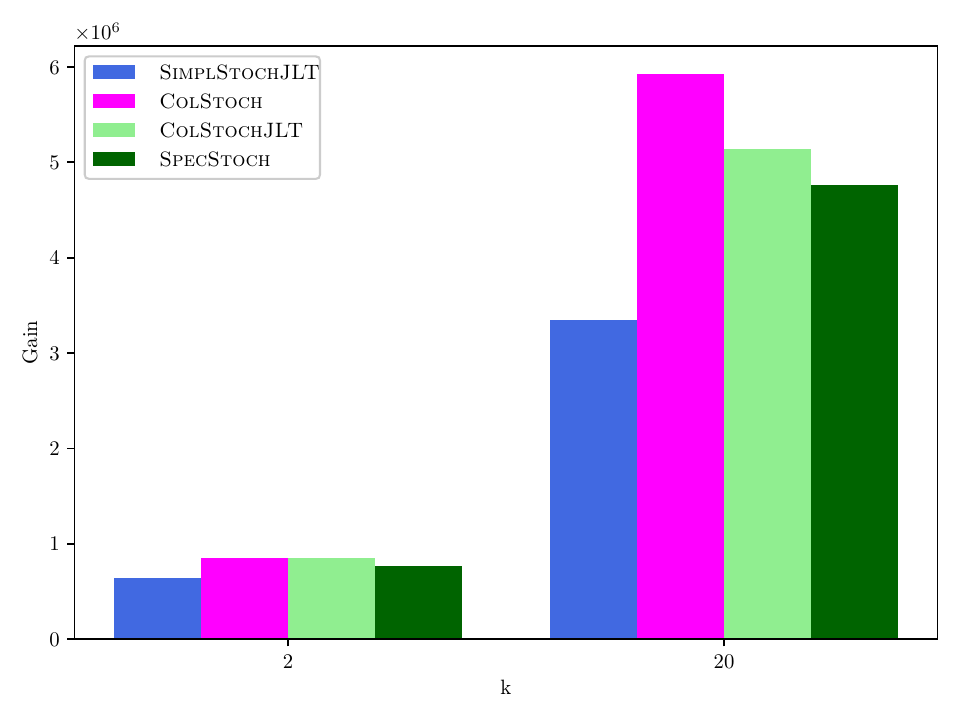}
  }
  \subfigure[Running time (log scale)]{
    \label{exp:large-time}
    \includegraphics[width=0.46\textwidth]{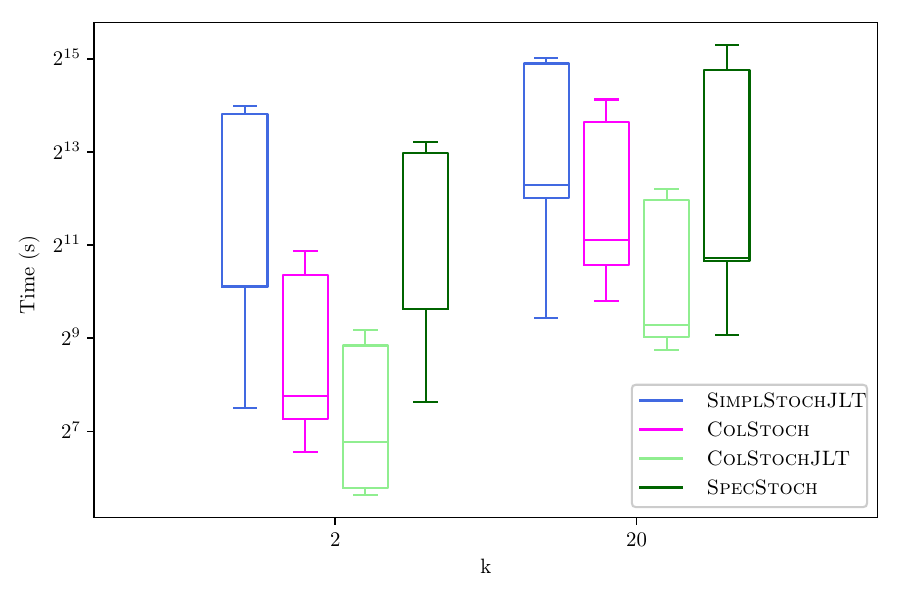}
  }
  \caption{Aggregated results (via geometric mean) of $k$-GRIP on large graphs ($n \geq 57K$) for different $k$.}
  \label{fig:large-comp}
\end{figure}

\subsection{Results for $k$-LRIP}
\extend{
  For $k$-LRIP, we use the same parameter settings determined in Section~\ref{sec:configuration-experiments}
  and choose $25$ focus nodes at random to run our algorithms on, with a $6$ hour time limit for each experiment.
  We evaluate $k\in\{2,5,20\}$, which means that we have to compute up to $20\cdot25 = 500$ \textsc{Update} steps overall -- up to $5\times$ more than for $k$-GRIP.
  At the same time, the number of \textsc{Eval} computations is reduced as described in Section~\ref{sec:lrip}.
  (One could of course increase $k$ and decrease the number of focus nodes at the same time and reach about the same number of \textsc{Update} and \textsc{Eval} calls.)

  Quality and speedup shown in this section are the geometric mean of the results for all focus nodes
  (in relation to the baseline \textsc{StGreedy}).
  Absolute running times are aggregated using the arithmetic mean.
  When comparing the running time of $k$-LRIP, we compute the running time for a focus node by taking the actual execution time of the main loop of Algorithm~\ref{alg:stochastic} for that focus node and add to this $\frac{1}{|F|}$ ($\frac{1}{25}$ in our case) of the preprocessing time, such that the preprocessing time is amortized over all focus nodes.

  For the evaluation we first compare the solution quality for the small and medium graphs of Table~\ref{tab:graphs}, see Figure~\ref{exp:lrip-smallmed-quality}.
  Here, \textsc{ColStoch} produces the best results, followed by \textsc{ColStochJLT}.
  Depending on $k$, \textsc{ColStoch} produces results that are on average 4\%-12\% away from \textsc{StGreedy}.
  The \textsc{SimplStoch*} results are 20\%-30\% away from \textsc{StGreedy}, showing that our graph-based sampling technique applied in \textsc{ColStoch} does improve the quality of the solution.
  \textsc{SpecStoch} appears to be not competitive.

  \begin{figure}[H]
    \centering
    \includegraphics[width=0.46\textwidth]{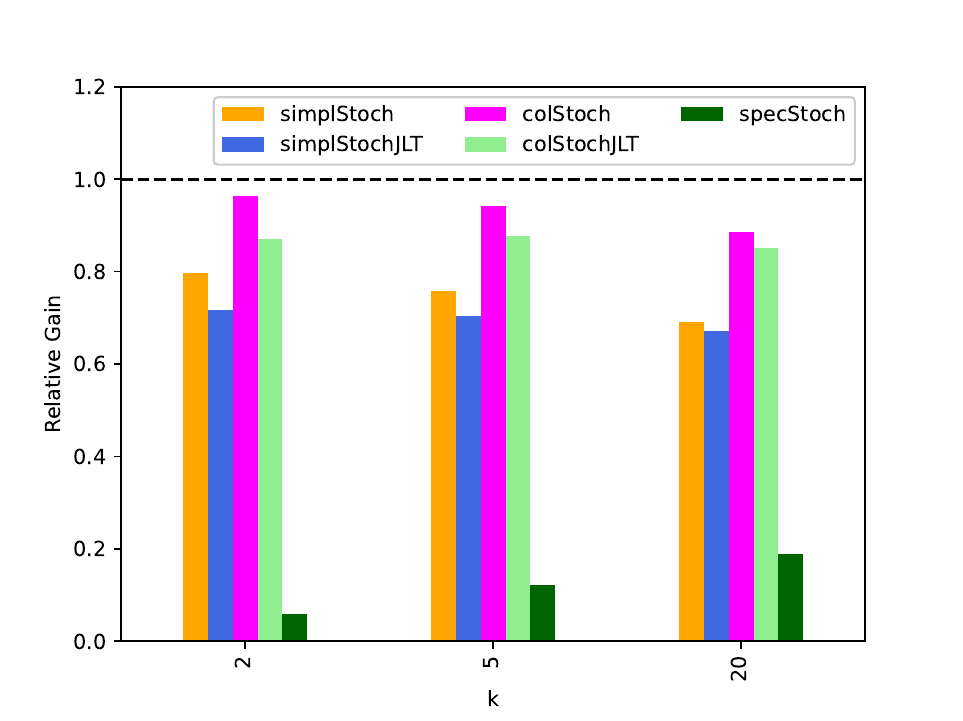}
    \caption{Aggregated quality results (using geometric mean) of $k$-LRIP on small and medium graphs ($n < 57K$) for different $k$. Results are relative to \textsc{StGreedy}. }
    \label{exp:lrip-smallmed-quality}
  \end{figure}

  Next, we take a look at the preprocessing time for our approaches (Figure~\ref{exp:lrip-preprocessing}).
  As expected, we can see a clear difference between the approaches that compute the full pseudoinverse (\textsc{stGreedy} and \textsc{simplStoch}) and those that set up a linear solver.
  The preprocessing time for the solver-based heuristics depends on the density of the graph.
  A good example of this observation is the difference in preprocessing time for the two instances \emph{web-indochina-2004} and \emph{arxiv-heph}.
  Both graphs have about the same number of nodes, but \emph{arxiv-heph} contains about 2.5x more edges, which increases the preprocessing time for \textsc{ColStoch}, \textsc{ColStochJLT} and \textsc{SimplStochJLT} by an order of magnitude.
  Still, the solver setup is considerably faster than calculating $\Lpinv$, being up to three orders of magnitude faster for larger, sparse graphs.
  Generally, the preprocessing of the \textsc{*JLT} variants is slightly slower than without JLT, since we have to setup the projection as well.
  Computing the eigenpairs for \textsc{SpecStoch} is faster than calculating $\Lpinv$, but slower than setting up linear solvers.
  One should keep in mind, though, that for \textsc{SpecStoch} this preprocessing computation is mostly the time to calculate the eigenpairs, which is the same computation required for the edge insertion update for \textsc{SpecStoch}.

  \begin{figure}[H]
    \centering
    \includegraphics[width=0.7\textwidth]{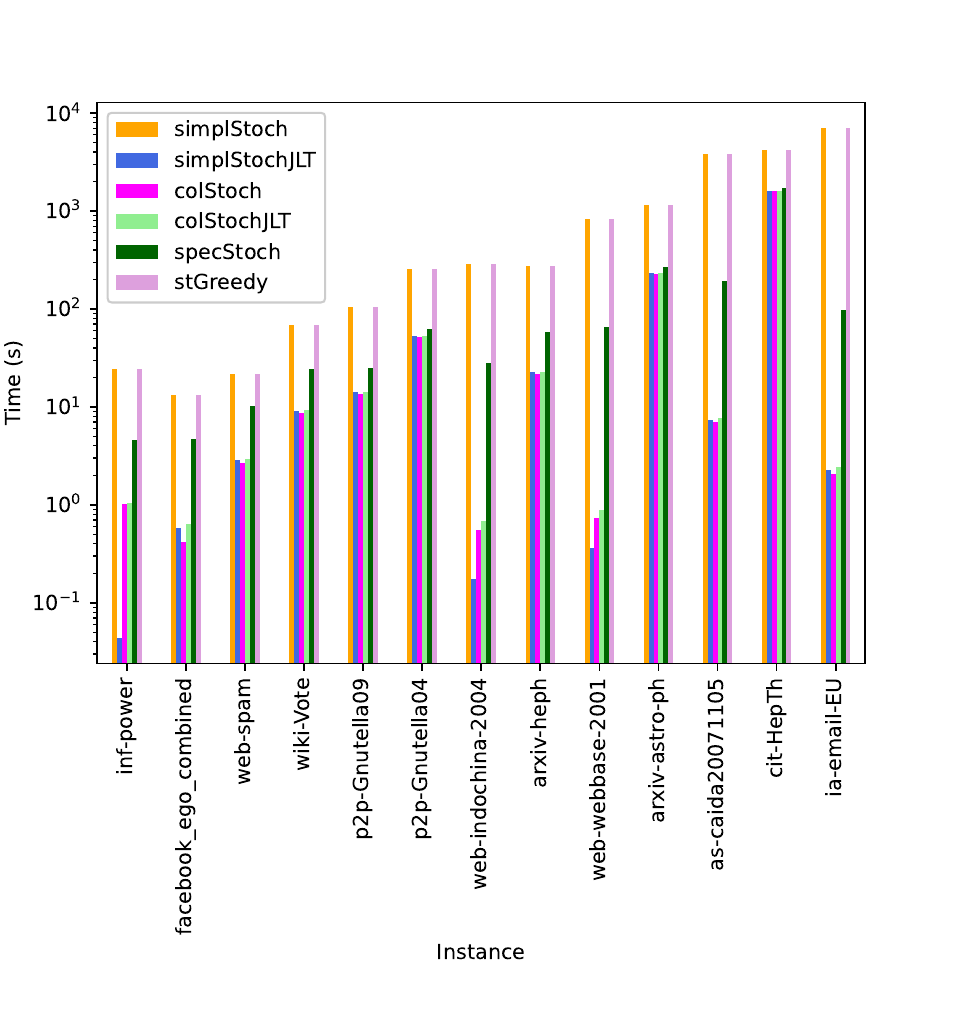}
    \caption{Preprocessing times for different graphs, taking the arithmetic mean over all $k$. See Table~\ref{tab:graphs} for size information. }
    \label{exp:lrip-preprocessing}
  \end{figure}

  Finally, we compare the running time of our approaches.
  We split the speedup results into two figures for small and medium graphs, respectively (Figure~\ref{exp:lrip-time-smallmed}).
  For both cases, \textsc{SpecStoch} has an average speedup of less than one.
  This is due to the large number of eigenpair computations required, which are slow, as we have seen in preprocessing.
  For this reason, most experiments with medium graphs and $k=20$ did not finish for \textsc{SpecStoch}.

  \begin{figure}[H]
    \centering
    \subfigure[Small Graphs ]{
      \label{exp:lrip-small-speedup}
      \includegraphics[width=0.46\textwidth]{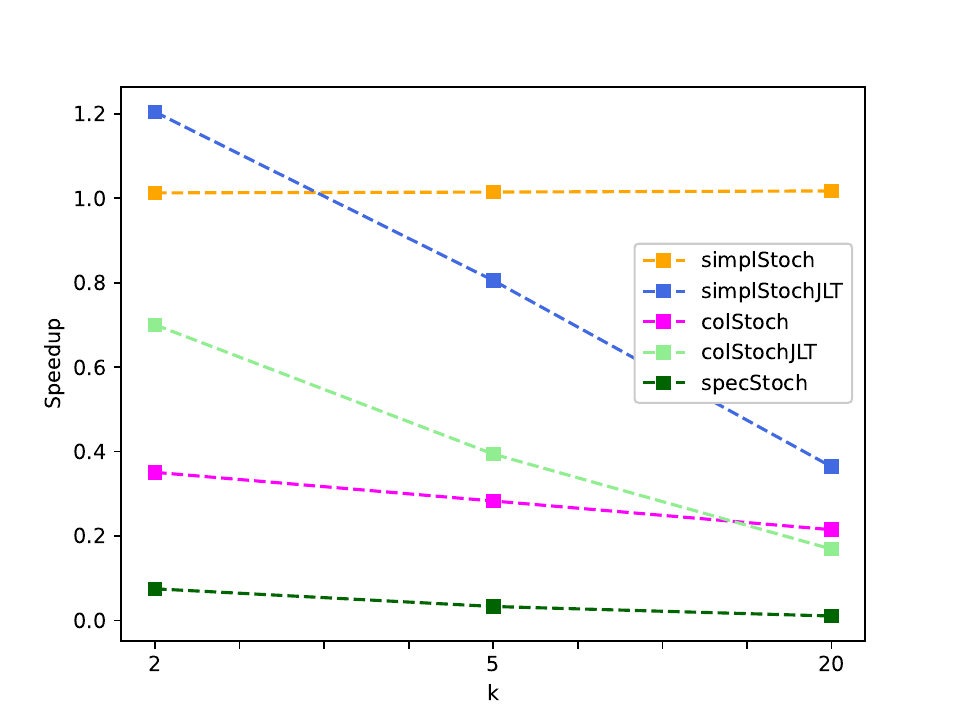}
    }
    \subfigure[Medium Graphs]{
      \label{exp:lrip-med-speedup}
      \includegraphics[width=0.46\textwidth]{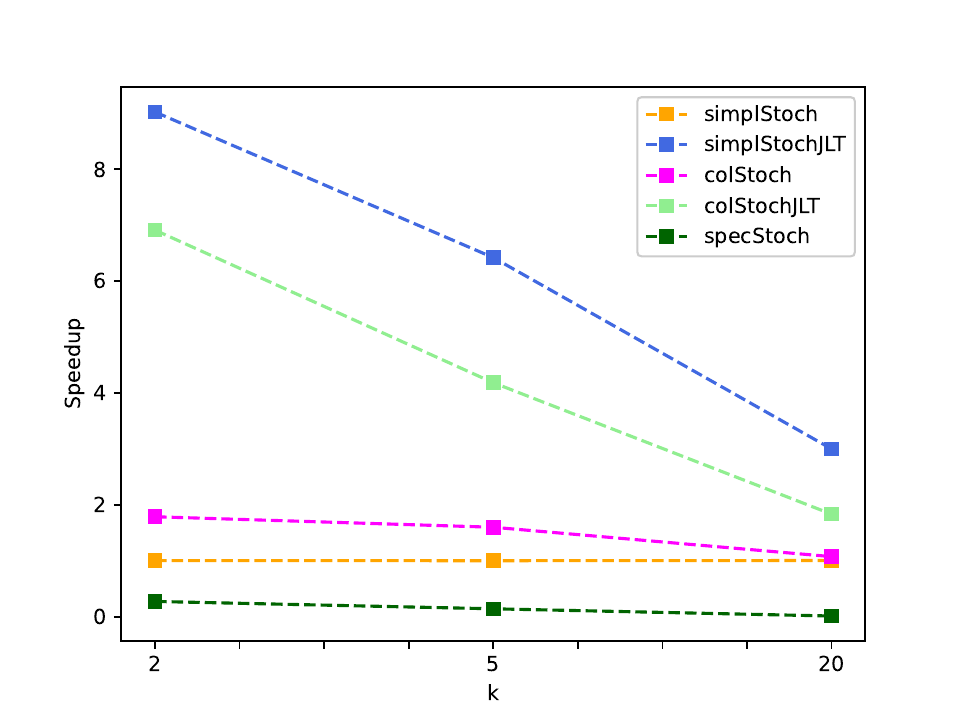}
    }
    \caption{Aggregated speedup results of $k$-LRIP on small and medium graphs for different $k$. Results are relative to \textsc{StGreedy}. }
    \label{exp:lrip-time-smallmed}
  \end{figure}

  Regarding the other heuristics: for small graphs, \textsc{SimplStoch} is the fastest algorithm, with an average speedup of $1.01$.
  The other algorithms are slower than \textsc{StGreedy}, because computing $\Lpinv$ for a small graph is still fast enough in practice and the update step generally is fast as well.
  Considering that all approaches finish in at most $12$ seconds (Figure~\ref{exp:lrip-small-absolute-time}),
  \textsc{StGreedy} is fast enough, so that these small graphs do not require (and do not benefit from) more complicated heuristics.

  For the medium graphs, \textsc{SimplStochJLT} is the fastest approach with a speedup of up to $9\times$ for $k=2$.
  This is to be expected, since the JLT strategy generally reduces computation time.
  The second fastest solution is \textsc{ColStochJLT}, which is explained by the additional time required to approximate $\diag\Lpinv$.
  \textsc{ColStoch} is still faster than \textsc{SimplStoch} for small $k$, but for $k=20$ both are almost equal.
  We also notice that for the \emph{cit-HepTh} graph, which is considerably denser than all other graphs ($m=352$K), the solver-based heuristics (\textsc{simplStochJLT}, \textsc{colStochJLT} and \textsc{colStoch}) time out, while the $\Lpinv$-based heuristics do not. The reason for this is that the time complexity of the solver update step depends on $m$.

  \begin{figure}[H]
    \centering
    \subfigure[Small Graphs ]{
      \label{exp:lrip-small-absolute-time}
      \includegraphics[width=0.46\textwidth]{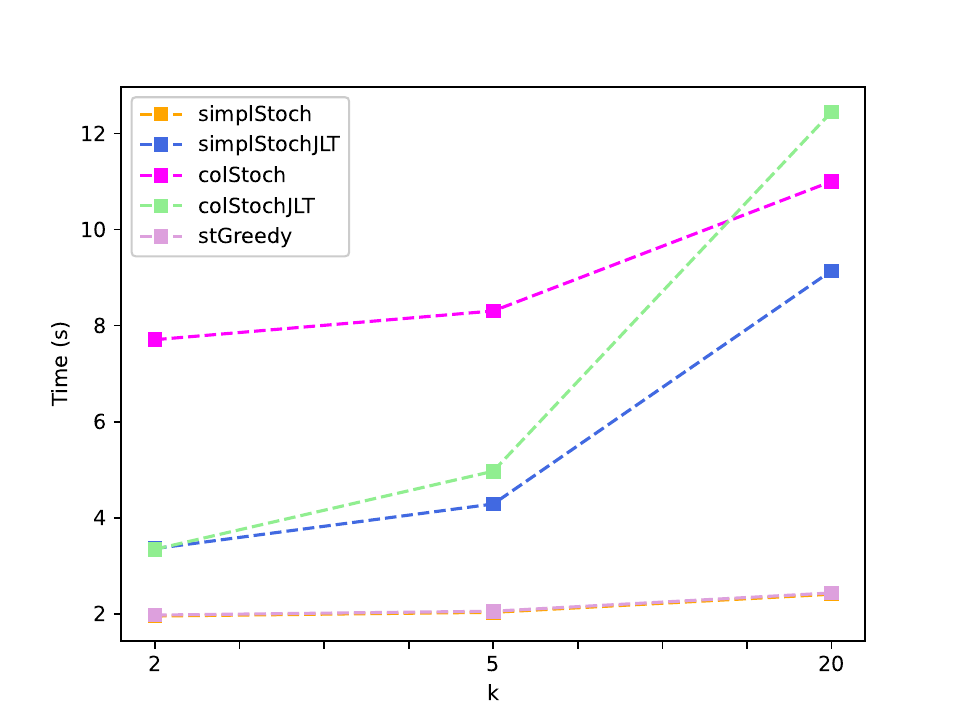}
    }
    \subfigure[Medium Graphs]{
      \label{exp:lrip-med-absolute-time}
      \includegraphics[width=0.46\textwidth]{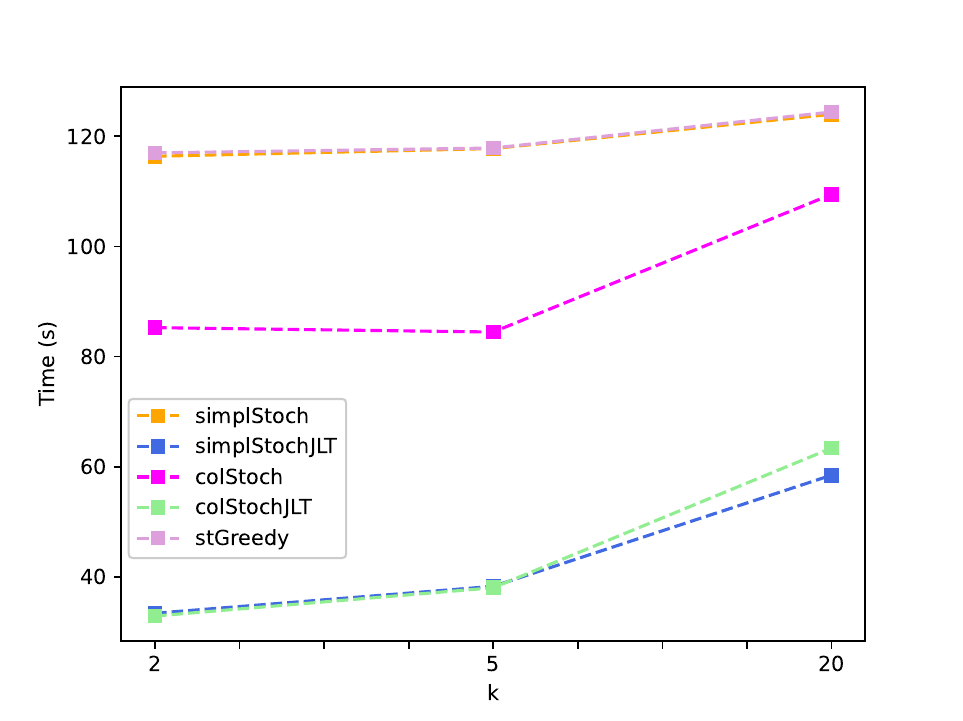}
    }
    \caption{Aggregated running time results of $k$-LRIP on small and medium graphs for different $k$. Results for \textsc{SpecStoch} are orders of magnitude larger and not shown here for readability. }
    \label{exp:lrip-time-smallmed-absoulte}
  \end{figure}

  Even though the preprocessing itself is more expensive for \textsc{SimplStoch}, once $\Lpinv$ is computed, the update step is considerably cheaper than in the case of linear solvers and as such \textsc{SimplStoch} is competitive for larger $k \cdot |F|$, where there are many updates, as long as computing $\Lpinv$ is feasible. Of course, for large enough graphs, one cannot compute $\Lpinv$ in reasonable time as we have seen for the large graphs in $k$-GRIP.

  Overall, based on these results the choice of the \emph{best} heuristic depends on $k$, $|F|$, and the density of the graph.
  In general, there is a trade-off between running time and quality. For the fastest solution, one should choose \textsc{SimplStochJLT}.
  When quality is the larger concern, \textsc{ColStoch} provides good results. With \textsc{ColStochJLT} there is also an option in the middle,
  providing good quality and time.

}

\section{Conclusions}
\label{sec:conc}

To conclude, our randomized techniques for speeding up the state-of-the-art greedy algorithm for $k$-GRIP
do pay off. For medium-sized graphs, {\sc ColStoch} provides already a decent $6\times$ acceleration with a quality close to greedy's.
Here, a subset of vertices $i$ is selected
for which $\ment{\Lpinv}{i}{i}$ and, thus, their summed effective resistances are large.
When favoring speed over quality, {\sc SpecStoch}, which exploits spectral properties of the graph,
offers an alternative (on average $28\times$ faster than greedy).
For larger graphs and whenever high quality is desirable,
the best option is {\sc ColStoch}.
When running time is important and a decrease in quality is allowed, {\sc ColStoch}
can still be significantly accelerated by JLT, i.\,e., {\sc ColStochJLT}.

\extend{Similar results can be observed for the related $k$-LRIP problem. Some differences occur, though: for small graphs (roughly \numprint{10000} nodes or less), \textsc{StGreedy} is fast
  enough since the running time and space consumption of the pseudoinversion is still tolerable and can be amortized over the numerous focus nodes.
  When the graphs become larger, our new heuristics pay off for $k$-LRIP as well -- except \textsc{SpecStoch},
  which is dominated in terms of quality \emph{and} running time.
}

Our future plans include the extension of the problem to edge deletions.
This problem is related to the protection of infrastructure and also important
in corresponding applications.

\paragraph*{Acknowledgments}
We are grateful for coding support in early development stages by HU Berlin student Matthias Görg.
\extend{Under the supervision of MP and HM, he also developed important ideas for the $\diag{\Lpinv}$ update strategy.}
\extend{Moreover, we thank Massimo Achterberg from Delft University of Technology for helpful discussions on several aspects of the paper.}

\bibliographystyle{abbrv}
\bibliography{references}

\end{document}